\documentclass[pra,aps,onecolumn,amsmath,amssymb,superscriptaddress,notitlepage,floatfix,12pt,showkeys]{revtex4}
\usepackage[dvipdfmx]{graphicx}
\usepackage{amsmath, amssymb, amsthm}
\usepackage{mathtools}
\usepackage{algorithm, algpseudocode}
\usepackage{color}
\usepackage{float}
\usepackage{bbm}

\setcitestyle{super}
\usepackage{etoolbox}
\makeatletter
\patchcmd{\NAT@citesuper}%
   {\textsuperscript{#1}}{\textsuperscript{[#1]}}{}{}
\makeatother

\def\Bra#1{\left\langle#1\right|}
\def\Ket#1{\left|#1\right\rangle}
\newcommand{\bra}[1]{\langle{#1}|}
\newcommand{\ket}[1]{|{#1}\rangle}

\DeclareMathOperator{\rank}{rank}
\DeclareMathOperator{\tr}{Tr}
\DeclareMathOperator{\id}{id}

\theoremstyle{definition}
\newtheorem{theorem}{Theorem}
\newtheorem{lemma}[theorem]{Lemma}

\newtheorem{proposition}[theorem]{Proposition}
\theoremstyle{definition}
\newtheorem{definition}[theorem]{Definition}

\newtheorem{application}{Application}
\theoremstyle{remark}

\begin{document}

\title{Distributed Encoding and Decoding of Quantum Information over Networks}

\author{\textit{Hayata Yamasaki*}}
\email{yamasaki@eve.phys.s.u-tokyo.ac.jp}
\affiliation{Department of Physics, Graduate School of Science, The University of Tokyo, 7--3--1 Hongo, Bunkyo-ku, Tokyo 113--0033, Japan}
\author{\textit{Mio Murao*}}
\email{murao@phys.s.u-tokyo.ac.jp}
\affiliation{Department of Physics, Graduate School of Science, The University of Tokyo, 7--3--1 Hongo, Bunkyo-ku, Tokyo 113--0033, Japan}

\begin{abstract}
   Encoding and decoding \textit{quantum information} in a multipartite quantum system are indispensable for quantum error correction and also play crucial roles in multiparty tasks in distributed quantum information processing such as quantum secret sharing. To quantitatively characterize nonlocal properties of multipartite quantum transformations for encoding and decoding, we analyze entanglement costs of encoding and decoding quantum information in a multipartite quantum system distributed among spatially separated parties connected by a network. This analysis generalizes previous studies of entanglement costs for preparing bipartite and multipartite quantum states and implementing bipartite quantum transformations by entanglement-assisted local operations and classical communication (LOCC). We identify conditions for the parties being able to encode or decode quantum information in the distributed quantum system deterministically and exactly, when inter-party quantum communication is restricted to a tree-topology network. In our analysis, we reduce the multiparty tasks of implementing the encoding and decoding to sequential applications of one-shot zero-error quantum state splitting and merging for two parties. While encoding and decoding are inverse tasks of each other, our results suggest that a quantitative difference in entanglement cost between encoding and decoding arises due to the difference between quantum state merging and splitting.
\end{abstract}


\maketitle

\section{\label{sec:intro}Introduction}

Encoding and decoding quantum information in a multipartite quantum system are fundamental building blocks in quantum information processing.
In particular, quantum error correcting codes~\cite{G,D,T2,B} require such encoding and decoding between a logical state and an entangled physical state of a multipartite system.
Quantum information is represented by this logical state,
and these encoding and decoding are the inverse transformation of each other, mathematically represented by isometries.
The encoding and decoding have to be performed so that coherence of these states is kept; that is, an arbitrary superposition of the logical state should be preserved without revealing the classical description of the logical state.
In addition to quantum information processing, the concept of encoding and decoding nowadays has interdisciplinary roles in analyzing many-body quantum systems exhibiting nonlocal features, such as topological order in quantum phase of matter,\cite{K,K2} holographic principle in quantum gravity,\cite{A,P} and eigenstate thermalization hypothesis in statistical physics.\cite{F}

\begin{figure}[b]
    \centering
    \includegraphics[width=0.9\linewidth]{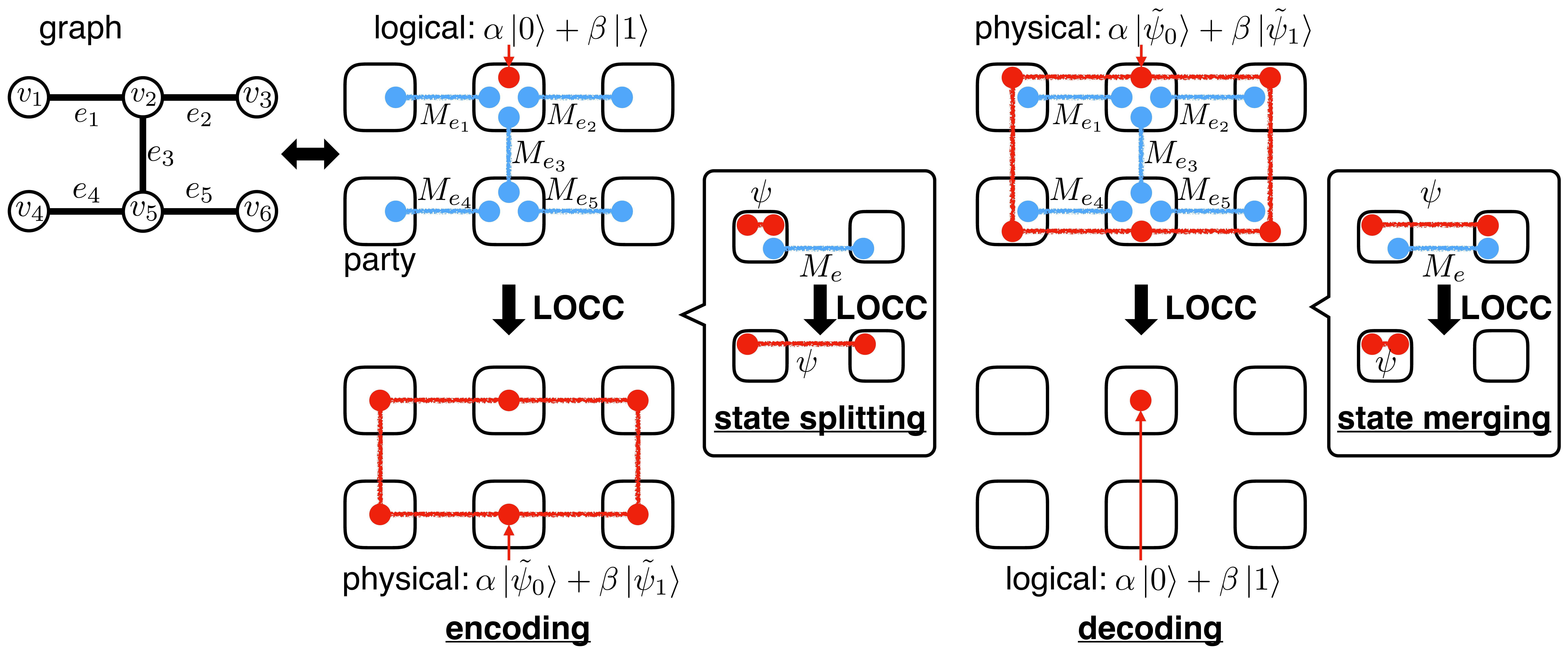}
    \caption{\label{fig:intro}Encoding and decoding quantum information in a multipartite quantum system shared among spatially separated parties, where the quantum information is represented by \textit{unknown} quantum states illustrated by red circles. The parties are connected by a network of noiseless quantum channels represented by a graph, so that the parties can sequentially apply exact state splitting to spread quantum information for encoding, and exact state merging to concentrate quantum information for decoding. Under LOCC, a single use of each noiseless quantum channel represented by an edge $e$ of the graph is equivalent to that of a maximally entangled state $\Ket{\Phi_{M_e}^+}$ illustrated by a pair of blue circles connected by a line, where $M_e$ is the Schmidt rank of $\Ket{\Phi_{M_e}^+}$.}
\end{figure}

These encoding and decoding are also indispensable when we aim to perform distributed quantum information processing, where spatially separated parties connected by a network for quantum communication cooperate in achieving an information processing task.
Distributed quantum information processing is considered to be a promising candidate for realizing large-scale quantum computation, since there exists technical difficulty in increasing the number of low-noise qubits in one quantum device.
Moreover, encoding and decoding are especially crucial for some multiparty cryptographic tasks such as quantum secret sharing.~\cite{B7,C,G2}
In these distributed settings, a multipartite system encoding a logical state is distributed among these spatially separated parties.
Thus, encoding and decoding are nonlocal transformations over all the parties, and the nonlocal properties of transformations for encoding and decoding lead to cost in implementations of the encoding and decoding.

This paper aims to quantitatively characterize nonlocal properties of transformations for encoding and decoding in the distributed settings, adopting a resource theoretical approach.
Resource theories~\cite{C8} aim to provide a quantitative understanding and an operational meaning of abstract physical attributes,
such as entanglement in quantum mechanics~\cite{H2,P2} and entropy in classical thermodynamics.\cite{L}
In the theory of entanglement, a class of operations, such as local operations and classical communication (LOCC), is adopted as free operations, to which available operations for achieving a given task are restricted.
Operations beyond LOCC have to be performed with assistance of nonlocal resource states, such as bipartite maximally entangled states.
Under LOCC, a single use of a noiseless quantum channel and that of a maximally entangled state are at equivalent cost by means of quantum teleportation achieving quantum communication.\cite{B5}
For a bipartite state, the minimal amount of quantum communication required for preparing the state provides a well-established entanglement measure quantifying a nonlocal property of the state, called the \textit{entanglement cost} of the state.\cite{B2,H,T}
The entanglement cost of a bipartite \textit{state} also generalizes to that required for spatially separated parties implementing a given nonlocal \textit{transformation}, such as global unitaries~\cite{Z,E,C3,C2,N,Y5,C4,Y,S,S2,Y4,Y2,S3,X,C5,V,Y3,W,W2,W5} and global measurements,\cite{J,B3,B4} although this generalization usually accompanies challenging optimization and has been analyzed only in special cases to date.
Another direction is generalization of a \textit{bipartite} state to a \textit{multipartite} state~\cite{Y6,G3,Y7} while analysis of multipartite entanglement is also challenging.\cite{E2,W3,B8}
Regarding a multipartite generalization in terms of quantum communication, Reference~\cite{Y6} formulates that required for preparing a multipartite state shared among parties using a network~\cite{K6} of the noiseless quantum channels.

Based on these previous works, we formalize entanglement costs characterizing the nonlocal properties of transformations for encoding and decoding.
In our setting, as illustrated in \textbf{Figure~\ref{fig:intro}}, $N$ parties are connected by a network of the noiseless quantum channels.
The network topology is represented by a graph in graph theory~\cite{B6} in terms of vertices and edges.
Any connected network of $N$ parties requires at least $N-1$ channels.
If an $N$-vertex connected graph has exactly $N-1$ edges, the graph is called a tree.
Using the network, the parties can spread and concentrate quantum information of \textit{unknown} states so as to encode and decode quantum information in a distributed system according to a given isometry representing the encoding and decoding.
The amount of quantum communication required for spreading and concentrating quantum information over the network characterizes nonlocal properties of the isometry.
Due to the equivalence between the noiseless quantum channel and the maximally entangled state,
a collection of maximally entangled states distributed according to the network topology comprises the initial resource state for spreading and concentrating quantum information by LOCC\@.
We assume that LOCC is free and consider this type of initial resource state consisting of \textit{bipartite} entanglement motivated by quantum communication on networks, while more general resource states exhibiting \textit{multipartite} entanglement may outperform bipartite entanglement in terms of other figures of merit than quantum communication.~\cite{Y12}
The minimal total amount of quantum communication is evaluated by the entanglement entropy of the maximally entangled states for each edge, which we call the \textit{entanglement costs of spreading and concentrating quantum information}.
The entanglement cost of spreading quantum information characterizes the encoding, and that of concentrating characterizes the decoding.

In this paper, we evaluate the entanglement costs of spreading and concentrating quantum information over any given tree-topology network for an \textit{arbitrarily} given isometry, which differs from the works presented in References~\cite{F3,S5} for implementing \textit{particular} isometries in the context of quantum secret sharing.
To analyze the entanglement costs,
we reduce spreading and concentrating quantum information to sequential applications of exact state merging and splitting for two parties established in Reference~\cite{Y11}, as illustrated in Figure~\ref{fig:intro}.
Regarding spreading quantum information, we use exact state splitting to provide an algorithm and derive the optimal entanglement cost of spreading quantum information, which is given in terms of the rank of a state defined with respect to each edge of the given tree.
We also provide another algorithm achieving concentrating quantum information using exact state merging, and show that the entanglement cost of concentrating quantum information can be reduced compared to that of spreading quantum information.
During spreading and concentrating \textit{quantum} information, coherence has to be kept, and this point is contrasted with encoding and decoding \textit{classical} information in quantum states shared among multiple parties investigated in the context of a type of quantum secret sharing based on LOCC state distinguishability.\cite{C6,R,Y10,W6,B11,L3}
Our algorithms for spreading and concentrating quantum information are applicable to any isometry representing encoding and decoding and provide an algorithm for one-shot distributed source compression~\cite{D8,D9,A8} applicable to arbitrarily-small-dimensional systems and a general algorithm for LOCC-assisted decoding of shared quantum information having studied in the context of quantum secret sharing.~\cite{G4}

The rest of this paper is organized as follows.
In Section~\ref{sec:preliminaries}, we introduce and formally define entanglement costs of spreading and concentrating quantum information over tree-topology networks, as well as summarizing the preceding results in Reference~\cite{Y11} on exact state merging and splitting.
We present our results on spreading quantum information over networks in Section~\ref{sec:encoding} and those on concentrating quantum information over networks in Section~\ref{sec:decoding}.
Applications of our results are provided in Section~\ref{sec:example}.
Our conclusion is given in Section~\ref{sec:conclusion}.

\section{\label{sec:preliminaries}Settings}
In this section, we define the tasks of spreading and concentrating quantum information over networks and the entanglement costs in Section~\ref{sec:task}.
We also summarize the results on exact state merging and splitting in Section~\ref{sec:merge_split_summary}.
In the following, for any finite-dimensional Hilbert space $\mathcal{H}$, the set of bounded operators on $\mathcal{H}$ is denoted by $\mathcal{B}\left(\mathcal{H}\right)$, and the set of density operators on $\mathcal{H}$ is denoted by $\mathcal{D}\left(\mathcal{H}\right)$.
We let superscripts of each ket or operator represent the systems to which the ket belongs or on which the operator acts.

\subsection{\label{sec:task}Entanglement cost of spreading and concentrating quantum information}

We introduce notations for describing the tasks on networks.
A quantum communication network among $N$ parties is represented by a graph $G=(V,E)$,
where each of the $N$ vertices $v\in V=\{v_1,v_2,\ldots,v_N\}$ represents one of the $N$ parties, and each edge $e=\{v_k,v_{k'}\}\in E$ represents a bidirectional noiseless quantum channel between $v_k$ and $v_{k'}$.
We assume that the $N$ parties can freely perform local operations and classical communication (LOCC).
Regarding a formal definition of LOCC, refer to Reference~\cite{C7}.
When LOCC is free, quantum communication of a state of an $M_e$-dimensional system is achieved by LOCC assisted by a maximally entangled state $\Ket{\Phi_{M_e}^+}^{e}\coloneqq\frac{1}{\sqrt{M_e}}\sum_{l=0}^{M_e-1}\Ket{l}^{v_k}\otimes\Ket{l}^{v_{k'}}$ of the Schmidt rank $M_e$ shared between $v_k$ and $v_{k'}$,
where the superscript $e=\{v_k,v_{k'}\}$ represents a state shared between $v_k$ and $v_{k'}$ connected by $e = \left\{v_k,v_{k'}\right\}$.
We consider the \textit{initial resource state} $\bigotimes_{e\in E}\Ket{\Phi_{M_e}^+}^e$ for $G=(V,E)$ as a resource for quantum communication on the network.

\begin{figure}[t]
    \centering
    \includegraphics[width=0.4\linewidth]{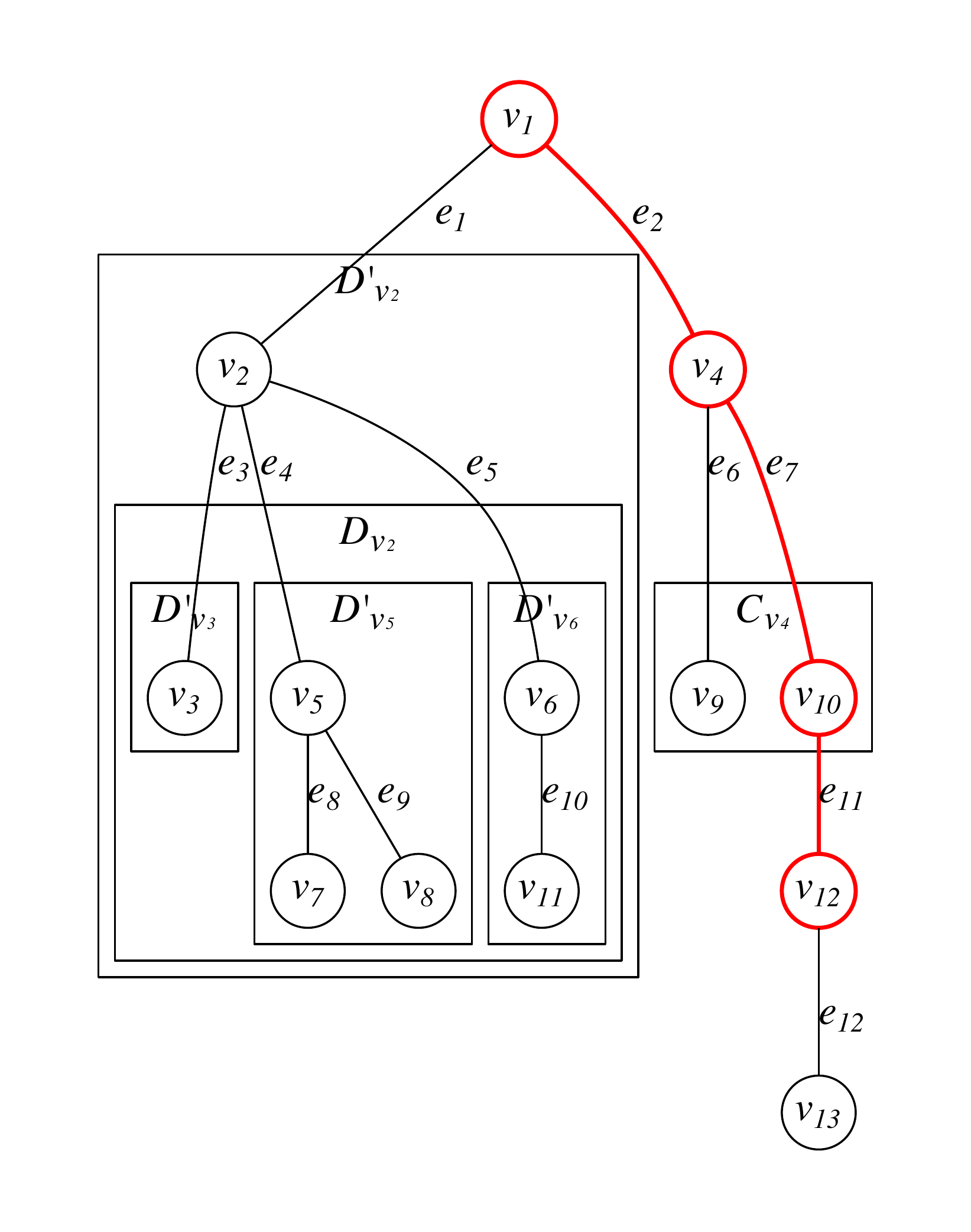}
    \caption{Notations for a tree $T=(V,E)$. We always designate $v_1\in V$ as the root of the tree. We label the other vertices so that, for any path connecting $v_1$ and another vertex, the closer to $v_1$ is any vertex $v_k$ on the path, the smaller is the label $k$, which we call an \textit{ascending labeling} of the vertices. For example, on the red bold path connecting the parties $v_1$ and $v_{12}$ in the figure, the vertices have to be labeled in ascending order $v_1,v_4,v_{10},v_{12}$. For any $v_k\in V$, let $C_{v_k}$, $D_{v_k}$, and $D'_{v_k}$ denote the set of $v_k$'s children, the set of $v_k$'s descendants, and the set of $v_k$ itself and $v_k$'s descendants, respectively.}
\label{fig:tree}
\end{figure}

Similarly to Reference~\cite{Y6},
we focus on analyzing networks represented by a class of graphs called \textit{trees},
which is defined as a graph containing no cycle as its subgraph.
We let $T=(V,E)$ denote a tree,
and our notations for trees are summarized in \textbf{Figure~\ref{fig:tree}}.
Given any tree $T=(V,E)$, we always designate $v_1\in V$ as the root of the tree.
In addition, we label the $N$ vertices $v_1,\ldots,v_N$ of $T$ so that, for any $v_k\neq v_1$, any vertex $v_{k'}$ on the path connecting the vertices $v_1$ and $v_k$ satisfies $k\geqq k'$; that is, the vertices are labeled in ascending order on such paths, as illustrated in Figure~\ref{fig:tree}.
We call this type of labeling of the vertices an \textit{ascending labeling} of the vertices.
For any vertex $v_k\in V$ of the rooted tree, we let $C_{v_k}$ denote the set of $v_k$'s children, $D_{v_k}$ the set of $v_k$'s descendants, and $D'_{v_k}$ the set of $v_k$ itself and $v_k$'s descendants.
For any non-root vertex $v_k\in V\setminus\left\{v_1\right\}$, we let $p\left(v_k\right)\in V$ denote the parent of $v_k$.
Any edge $e\in E$ of the rooted tree can be written as $e=\{p\left(v_k\right),v_k\}$.

Given a network represented by any graph $G=(V,E)$ in general,
we define the tasks of spreading and concentrating quantum information according to a given isometry representing encoding or decoding.
A system $\mathcal{H}$ for logical states is located at one of the $N$ parties,
and we always assign $v_1\in V$ as the party where $\mathcal{H}$ is located.
We let $D$ denote the dimension of $\mathcal{H}$, that is, $D=\dim\mathcal{H}$.
We write the computational basis of $\mathcal{H}$ as
$\{\Ket{l}\in\mathcal{H}:l=0,1,\ldots,D-1\}$.
In addition, the $N$ parties share a multipartite system $\tilde{\mathcal{H}}$ for physical states.
The system $\tilde{\mathcal{H}}$ is required to be spanned by a set of $D$ orthonormal pure states
$\left\{\ket{\tilde{\psi}_l}^{v_1\cdots v_N}\in\tilde{\mathcal{H}}:l=0,1,\ldots,D-1\right\}$.
For each $v_k\in V$, we let $\tilde{\mathcal{H}}^{v_k}$ denote a part of the shared multipartite system $\tilde{\mathcal{H}}$ located at the party $v_k$.
Note that $\dim\tilde{\mathcal{H}}^{v_k}$ is arbitrary as long as satisfying $\dim\tilde{\mathcal{H}}=\dim\mathcal{H}=D$,
and hence, $\tilde{\mathcal{H}}$ is a \textit{subspace} of the Hilbert space consisting of these subsystems for the $N$ parties, that is,
$\tilde{\mathcal{H}}\subset\bigotimes_{v\in V}\tilde{\mathcal{H}}^{v}$.
We consider encoding and decoding as linear bijective maps between $\mathcal{B}\left(\mathcal{H}\right)$ and $\mathcal{B}\left(\tilde{\mathcal{H}}\right)$ mapping the basis states of $\mathcal{H}$ and $\tilde{\mathcal{H}}$ as
$\ket{l}\in\mathcal{H}\leftrightarrow\ket{\tilde{\psi}_l}\in\tilde{\mathcal{H}}$ for each $l\in\left\{0,\ldots,D-1\right\}$.
The encoding map is represented by an isometry $U$ from $\mathcal{H}$ to $\tilde{\mathcal{H}}$ satisfying
$\ket{\tilde{\psi}_l}=U\Ket{l}$.
Encoding refers to transformation from $\rho\in\mathcal{D}\left(\mathcal{H}\right)$ into $U\rho U^\dag\in\mathcal{D}\left(\tilde{\mathcal{H}}\right)$, and decoding refers to the inverse transformation represented by $U^\dag$.
The formal definitions of the tasks of spreading and concentrating quantum information are given in terms of the LOCC framework as follows.
Note that the tasks are performed deterministically and exactly.
\begin{definition}
    \textit{Spreading and concentrating quantum information.}
    Spreading quantum information over a given graph $G=(V,E)$
    for a given isometry $U$
    is a task for the $N$ parties $v_1,\ldots,v_N\in V$ to apply $U$ to an arbitrary \textit{unknown} input state $\rho\in\mathcal{D}\left(\mathcal{H}\right)$ of one party $v_1\in V$ to share $U\rho U^\dag\in\mathcal{D}\left(\tilde{\mathcal{H}}\right)$ among the $N$ parties by performing an LOCC map $\mathcal{S}$ assisted by an initial resource state $\bigotimes_{e\in E}\Ket{\Phi_{M_e}^+}^e$, that is,
    \begin{equation}
        \label{eq:encoding}
        \mathcal{S}\left(\rho\otimes\bigotimes_{e\in E}\Ket{\Phi_{M_e}^+}\Bra{\Phi_{M_e}^+}\right)=U\rho U^\dag.
    \end{equation}
    Concentrating quantum information over $G=(V,E)$
    for $U$
    is a task for the $N$ parties $v_1,\ldots,v_N\in V$ to apply $U^\dag$ to a shared input state $U\rho U^\dag\in\mathcal{D}\left(\tilde{\mathcal{H}}\right)$ corresponding to an arbitrary \textit{unknown} state $\rho\in\mathcal{D}\left(\mathcal{H}\right)$ to recover $\rho$ at one party $v_1\in V$ by performing an LOCC map $\mathcal{C}$ assisted by $\bigotimes_{e\in E}\Ket{\Phi_{M_e}^+}^e$, that is,
    \begin{equation}
        \label{eq:decoding}
        \mathcal{C}\left(U\rho U^\dag\otimes\bigotimes_{e\in E}\Ket{\Phi_{M_e}^+}\Bra{\Phi_{M_e}^+}\right)=\rho.
    \end{equation}
\end{definition}

We analyze minimum requirements for initial resource states achieving spreading and concentrating quantum information.
In the same way as the case of analyzing the entanglement cost of multipartite states construction,~\cite{Y6}
given any graph $G=(V,E)$,
the entanglement cost of consuming the bipartite maximally entangled state $\Ket{\Phi^+_{M_e}}^e$ for each $e\in E$ of the initial resource state $\bigotimes_{e\in E}\Ket{\Phi^+_{M_e}}$ is identified by the entanglement entropy of $\Ket{\Phi^+_{M_e}}^e$, that is, $\log_2 M_e$.
If a sufficiently large amount of entanglement is available for each edge, there exist trivial algorithms for achieving spreading and concentrating quantum information, simply using quantum teleportation~\cite{B5} so that the party $v_1$ can locally perform any given isometry on the unknown input state.
In contrast, we aim to reduce the total amount of entanglement $\sum_{e\in E}\log_2 M_e$ required for spreading and concentrating quantum information.
Note that $\sum_{e\in E}\log_2 M_e$ can be regarded as the total amount of quantum communication available on the network when LOCC can be freely performed, as explained in Reference~\cite{Y6}.

\begin{definition}
\textit{Entanglement costs of spreading and concentrating quantum information.}
The \textit{entanglement cost of spreading quantum information} over a given graph $G=(V,E)$ for a given isometry $U$ is a family
${\left(\log_2 M_e\right)}_{e\in E}$
identifying an initial resource state achieving spreading quantum information over $G$ for $U$ minimizing
$\sum_{e\in E}\log_2 M_e$.
The \textit{entanglement cost of concentrating quantum information} over $G$ for $U$ is a family
${\left(\log_2 M_e\right)}_{e\in E}$
identifying an initial resource state achieving concentrating quantum information over $G$ for $U$ minimizing
$\sum_{e\in E}\log_2 M_e$.
\end{definition}

\begin{figure}[t]
    \centering
    \begin{minipage}[t]{0.45\hsize}
        \centering
        \includegraphics[width=0.9\linewidth]{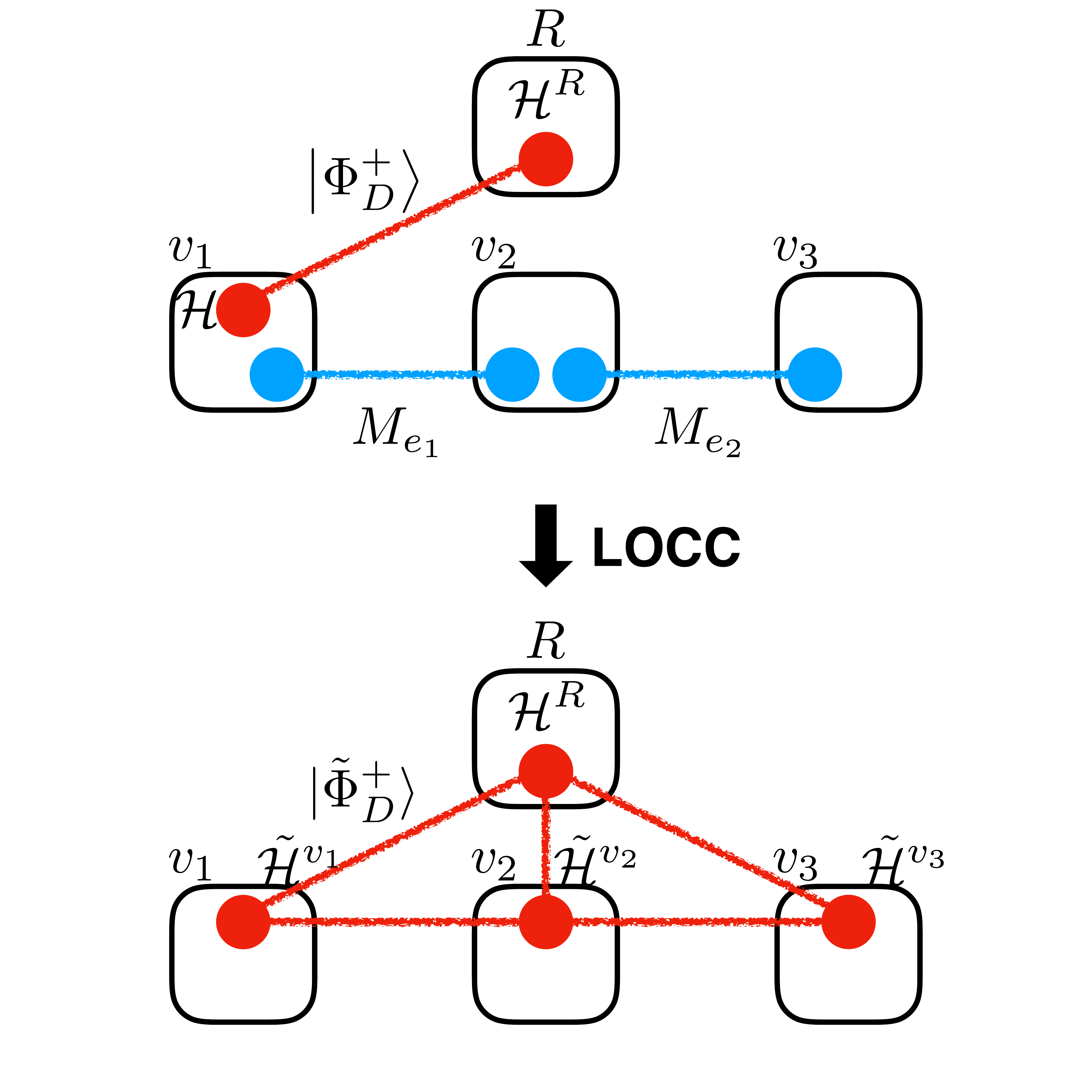}
        \caption{\label{fig:encoding}A state transformation task for three parties $v_1$, $v_2$, and $v_3$ equivalent to spreading quantum information over a line-topology network, where the system $\mathcal{H}$ for logical states is located at $v_1$. The initial state $\Ket{\Phi_D^+}$ and the final state $\ket{\tilde{\Phi_D^+}}$ are defined as Equations~\eqref{eq:data_maximally_entangled_state} and~\eqref{eq:code_maximally_entangled_state}, respectively.}
    \end{minipage}
    \begin{minipage}[t]{0.08\hsize}
        \hspace{2mm}
    \end{minipage}
    \begin{minipage}[t]{0.45\hsize}
        \centering
        \includegraphics[width=0.9\linewidth]{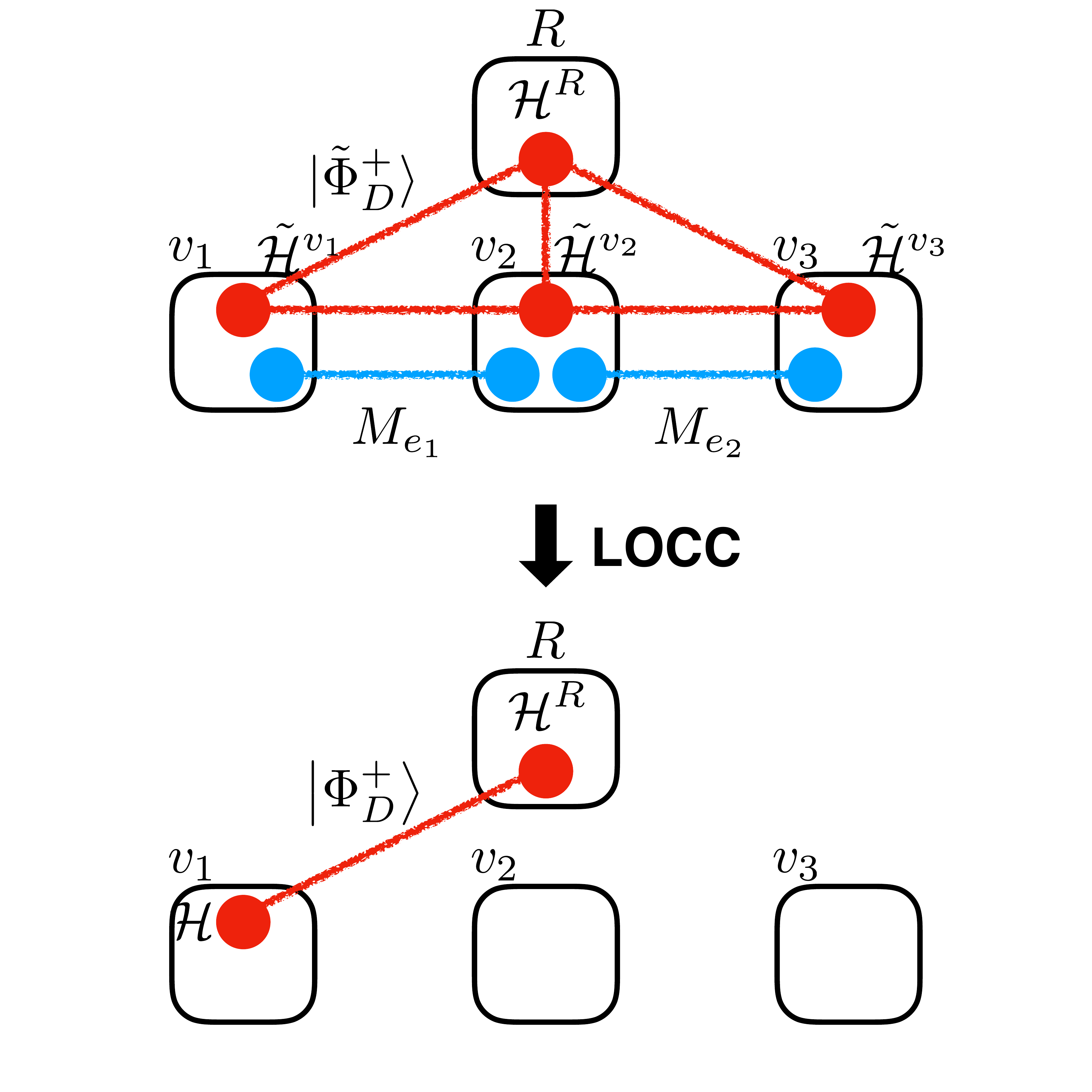}
        \caption{\label{fig:decoding}A state transformation task equivalent to concentrating quantum information over the same network as Figure~\ref{fig:encoding}. The notations are the same as those in Figure~\ref{fig:encoding}.}
    \end{minipage}
\end{figure}

To analyze the entanglement costs of spreading and concentrating quantum information, we reduce these tasks to a particular type of state transformations.
Given any graph $G=(V,E)$ and any isometry $U$,
the state transformation equivalent to spreading quantum information over $G$ for $U$ is illustrated in \textbf{Figure~\ref{fig:encoding}}, and the state transformation equivalent to concentrating in \textbf{Figure~\ref{fig:decoding}}.
To define these equivalent state transformations, we consider a $D$-dimensional system $\mathcal{H}^R$ located at a party $R$ other than the $N$ parties $v_1,\ldots,v_N\in V$, where $\mathcal{H}^R$ is a reference system on which none of the $N$ parties can apply any operation.
Note that $D=\dim\mathcal{H}$.
We write a maximally entangled state of the Schmidt rank $D$ shared between $R$ and $v_1$ as
\begin{equation}
    \label{eq:data_maximally_entangled_state}
    \Ket{\Phi^+_D}=\frac{1}{\sqrt{D}}\sum_{l=0}^{D-1}\Ket{l}\otimes\Ket{l}\in\mathcal{H}^R\otimes\mathcal{H}.
\end{equation}
Moreover, we write a state obtained by performing $U$ on $\mathcal{H}$ for $\Ket{\Phi^+_D}$ as
\begin{equation}
    \label{eq:code_maximally_entangled_state}
    \begin{split}
        \ket{\tilde{\Phi}^+_D}\coloneqq\left(\mathbbm{1}^R\otimes U\right)\Ket{\Phi^+_D}
                              =\frac{1}{\sqrt{D}}\sum_{l=0}^{D-1}\Ket{l}\otimes\ket{\tilde{\psi}_l}\in\mathcal{H}^R\otimes\tilde{\mathcal{H}},
    \end{split}
\end{equation}
where $\mathbbm{1}^R$ is the identity operator on the system $\mathcal{H}^R$.
The equivalence between the two tasks is stated as follows,
which is a straightforward generalization of a technique referred to in Reference~\cite{P3} as the relative state method.

\begin{proposition}
\label{lem:encoding_state_transformation}
    \textit{State transformations equivalent to spreading and concentrating quantum information.}
    Spreading quantum information over a given graph $G=(V,E)$ for a given isometry $U$ defined as Equation~\eqref{eq:encoding} is achievable if and only if there exists an LOCC map $\mathcal{S}$ by the $N$ parties assisted by the initial resource state $\bigotimes_{e\in E}\Ket{\Phi_{M_e}^+}^e$ such that
    \begin{equation}
        \label{eq:encoding_state_transformation}
        \begin{split}
            \id^R\otimes\mathcal{S}\left(\Ket{\Phi^+_D}\Bra{\Phi^+_D}\otimes\bigotimes_{e\in E}\Ket{\Phi_{M_e}^+}\Bra{\Phi_{M_e}^+}\right)
            =\ket{\tilde\Phi^+_D}\bra{\tilde\Phi^+_D},
        \end{split}
    \end{equation}
    where $\id^R$ is the identity map on $\mathcal{H}^R$, and the states $\Ket{\Phi^+_D}$ and $\ket{\tilde\Phi^+_D}$ are defined as Equations~\eqref{eq:data_maximally_entangled_state} and~\eqref{eq:code_maximally_entangled_state}, respectively.
    Concentrating quantum information over $G=(V,E)$ for $U$ defined as Equation~\eqref{eq:decoding} is achievable if and only if there exists an LOCC map $\mathcal{C}$ by the $N$ parties assisted by $\bigotimes_{e\in E}\Ket{\Phi_{M_e}^+}^e$ such that
    \begin{equation}
        \label{eq:decoding_state_transformation}
        \begin{split}
            \id^R\otimes\mathcal{C}\left(\ket{\tilde\Phi^+_D}\bra{\tilde\Phi^+_D}\otimes\bigotimes_{e\in E}\Ket{\Phi_{M_e}^+}\Bra{\Phi_{M_e}^+}\right)
            =\Ket{\Phi^+_D}\Bra{\Phi^+_D},
        \end{split}
    \end{equation}
    where the notations are the same as those in Equation~\eqref{eq:encoding_state_transformation}.
\end{proposition}

\begin{proof}
   We prove the statement on spreading quantum information, while the statement on concentrating quantum information also follows from the same argument by substituting $\rho$, $U\rho U^\dag$, $\Ket{l}\Bra{l'}$, $\ket{\tilde{\psi}_l}\bra{\tilde{\psi}_{l'}}$, and $\mathcal{S}$ in the following with $U\rho U^\dag$, $\rho$, $\ket{\tilde{\psi}_l}\bra{\tilde{\psi}_{l'}}$, $\Ket{l}\Bra{l'}$, and $\mathcal{C}$, respectively.

    \textit{If part}:
    If there exists an LOCC map $\mathcal{S}$ defined as Equation~\eqref{eq:encoding} for any input state $\rho$,
    Equation~\eqref{eq:encoding_state_transformation} holds as a special case of Equation~\eqref{eq:encoding} in which the input state $\rho$ is a completely mixed state.

    \textit{Only if part}:
    Assume that there exists an LOCC map $\mathcal{S}$ defined as Equation~\eqref{eq:encoding_state_transformation}.
    Due to the linearity of the map $\mathcal{S}$, Equation~\eqref{eq:encoding_state_transformation} yields
    \begin{equation}
        \begin{split}
            \frac{1}{D}\sum_{l,l'=0}^{D-1}\Ket{l}\Bra{l'}\otimes\mathcal{S}\left(\Ket{l}\Bra{l'}\otimes\bigotimes_{e\in E}\Ket{\Phi_{M_e}^+}\Bra{\Phi_{M_e}^+}\right)
            =\frac{1}{D}\sum_{l,l'=0}^{D-1}\Ket{l}\Bra{l'}\otimes\ket{\tilde{\psi}_l}\bra{\tilde{\psi}_{l'}}.
        \end{split}
    \end{equation}
    Since the set ${\left\{\Ket{l}\Bra{l'}\right\}}_{l,l'}$ of operators on the system $\mathcal{H}^R$ is linearly independent, we obtain
    \begin{equation}
        \mathcal{S}\left(\Ket{l}\Bra{l'}\otimes\bigotimes_{e\in E}\Ket{\Phi_{M_e}^+}\Bra{\Phi_{M_e}^+}\right) = \ket{\tilde{\psi}_l}\bra{\tilde{\psi}_{l'}},
    \end{equation}
    for each $l,l'\in\left\{0,\ldots,D-1\right\}$.
    Therefore, writing any operators $\rho\in\mathcal{D}\left(\mathcal{H}\right)$ and $U\rho U^\dag\in\mathcal{D}\left(\tilde{\mathcal{H}}\right)$ as
    \begin{equation}
        \rho=\sum_{l,l'=0}^{D-1}c_{l,l'}\Ket{l}\Bra{l'},\quad U\rho U^\dag=\sum_{l,l'=0}^{D-1}c_{l,l'}\ket{\tilde{\psi}_l}\bra{\tilde{\psi}_{l'}},
    \end{equation}
    we conclude Equation~\eqref{eq:encoding} as follows:
    \begin{equation}
        \begin{split}
            \mathcal{S}\left(\rho\otimes\bigotimes_{e\in E}\Ket{\Phi_{M_e}^+}\Bra{\Phi_{M_e}^+}\right)
            =\sum_{l,l'=0}^{D-1}c_{l,l'}\mathcal{S}\left(\Ket{l}\Bra{l'}\otimes\bigotimes_{e\in E}\Ket{\Phi_{M_e}^+}\Bra{\Phi_{M_e}^+}\right)
            =\sum_{l,l'=0}^{D-1}c_{l,l'}\ket{\tilde{\psi}_l}\bra{\tilde{\psi}_{l'}}
            =U\rho U^\dag.
        \end{split}
    \end{equation}
\end{proof}

\subsection{\label{sec:merge_split_summary}Exact state splitting and merging}
We summarize the results on exact state splitting and merging presented in Reference~\cite{Y11} for later use in our analysis.
Consider three parties $A$, $B$, and $R^\prime$.
The parties $A$ and $B$ may freely perform LOCC but cannot perform any operation on the system of the party $R^\prime$.

Regarding exact state splitting, Reference~\cite{Y11} provides an explicit construction of an algorithm for exact state splitting, which yields the following lemma.
\begin{lemma}
\label{lem:split}
(Reference~\cite{Y11})
\textit{Exact state splitting.}
Given any pure state $\Ket{\psi}^{R^\prime AA^\prime}\in\mathcal{H}^{R^\prime}\otimes\mathcal{H}^A\otimes\mathcal{H}^{A^\prime}$, where $\mathcal{H}^{R^\prime}$ is the Hilbert space of $R^\prime$, and $\mathcal{H}^A\otimes\mathcal{H}^{A^\prime}$ of $A$,
there exists an LOCC map performed by $A$ and $B$ transforming
$\Ket{\psi}^{R^\prime AA^\prime}\otimes\Ket{\Phi_K^+}$ to $\Ket{\psi}^{R^\prime AB}$,
where $\Ket{\psi}^{R^\prime AB}\in\mathcal{H}^{R^\prime}\otimes\mathcal{H}^A\otimes\mathcal{H}^B$, $\mathcal{H}^B$ is the Hilbert space of $B$, and $\Ket{\Phi_K^+}$ is a maximally entangled state shared between $A$ and $B$ satisfying
\begin{equation}
    \log_2 K = \log_2\rank\tr_{R^\prime A}\Ket{\psi}\Bra{\psi}^{R^\prime AA^\prime}.
\end{equation}
\end{lemma}

For exact state merging, the following Koashi-Imoto decomposition~\cite{K3,H6,K5,W4} is a key technique for evaluating the entanglement cost.
\begin{lemma}
\label{lem:koashi_imoto_decomposition_tripartite}
(Lemma~11 in Reference~\cite{W4}, Reference~\cite{Y11})
\textit{Koashi-Imoto decomposition of a tripartite pure state.}
Given any pure state $\Ket{\psi}^{R^\prime AB}\in\mathcal{H}^{R^\prime}\otimes\mathcal{H}^A\otimes\mathcal{H}^B$,
there exists an algorithmic procedure to obtain a decomposition of $\mathcal{H}^A$ and $\mathcal{H}^B$
\begin{equation}
    \mathcal{H}^A=\bigoplus_{j=0}^{J-1}\mathcal{H}^{a_j^L}\otimes\mathcal{H}^{a_j^R},\quad
    \mathcal{H}^B=\bigoplus_{j=0}^{J-1}\mathcal{H}^{b_j^L}\otimes\mathcal{H}^{b_j^R},
\end{equation}
such that
$\Ket{\psi}^{R^\prime AB}$ is decomposed into
\begin{equation}
    \Ket{\psi}^{R^\prime AB}=\bigoplus_{j=0}^{J-1}\sqrt{p\left(j\right)}\Ket{\omega_j}^{a_j^L b_j^L}\otimes\Ket{\phi_j}^{R^\prime a_j^R b_j^R},
\end{equation}
where $p\left(j\right)$ is a probability distribution, and, for each $j\in\{0,\ldots,J-1\}$, $\Ket{\omega_j}^{a_j^L b_j^L}\in\mathcal{H}^{a_j^L}\otimes\mathcal{H}^{b_j^L}$ and $\Ket{\phi_j}^{R^\prime  a_j^R b_j^R}\in\mathcal{H}^{R^\prime}\otimes\mathcal{H}^{a_j^R}\otimes\mathcal{H}^{b_j^R}$.
\end{lemma}

Using the Koashi-Imoto decomposition, an explicit construction of an algorithm for exact state merging can be obtained, which yields the following lemma.
This algorithm consists of $A$'s projective measurement ${\left\{\Ket{m}^A\right\}}_m$ and $B$'s isometry $U_m^B$ conditioned by $m$, and their explicit forms are given in Reference~\cite{Y11}.
\begin{lemma}
\label{lem:merge}
(Reference~\cite{Y11})
\textit{Exact state merging.}
Given any pure state $\Ket{\psi}^{R^\prime AB}\in\mathcal{H}^{R^\prime}\otimes\mathcal{H}^A\otimes\mathcal{H}^B$, where $\mathcal{H}^{R^\prime}$, $\mathcal{H}^A$, and $\mathcal{H}^B$ are the Hilbert spaces of $R^\prime$, $A$, and $B$, respectively,
there exists an LOCC map performed by $A$ and $B$ achieving
\begin{equation}
    \label{eq:locc_merge}
    \left(\mathbbm{1}^{R^\prime}\otimes\Bra{m}^A\otimes U_m^B\right)\left(\Ket{\psi}^{R^\prime AB}\otimes\Ket{\Phi_K^+}\right)=\Ket{\psi}^{R^\prime B^\prime B},
\end{equation}
where $\mathbbm{1}^{R^\prime}$ is the identity operator on the system $\mathcal{H}^{R^\prime}$, ${\left\{\Ket{m}^A\right\}}_m$ represents a projective measurement by $A$ with the outcome $m$, $U_m^B$ represents an isometry by $B$ conditioned by $m$, $\Ket{\psi}^{R^\prime B^\prime B}\in\mathcal{H}^{R^\prime}\otimes\mathcal{H}^{B^\prime}\otimes\mathcal{H}^B$, $\mathcal{H}^{B^\prime}$ is $B$'s system satisfying $\dim\mathcal{H}^A=\dim\mathcal{H}^{B^\prime}$, and $\Ket{\Phi_K^+}$ is a maximally entangled state shared between $A$ and $B$ satisfying
\begin{equation}
        \log_2 K = \max_{j}\left\{\log_2\left\lceil\lambda_0^{a_j^L}\dim\mathcal{H}^{a_j^R}\right\rceil\right\},
\end{equation}
where $j\in\{0,\ldots,J-1\}$, $\lambda^{a_j^L}_0$ is the largest eigenvalue of the reduced state $\omega_j^{a_j^L}\coloneqq\tr_{b_j^L}\Ket{\omega_j}\Bra{\omega_j}^{a_j^L b_j^L}\in\mathcal{D}\left(\mathcal{H}^{a_j^L}\right)$ of $\Ket{\omega_j}^{a_j^L b_j^L}$, and $\lceil{}\cdots{}\rceil$ is the ceiling function.
\end{lemma}

\section{\label{sec:encoding}Entanglement cost of spreading quantum information}
We derive the optimal entanglement cost of spreading quantum information over any tree for any isometry.
To derive the entanglement cost, we generalize the two-party algorithm for exact state splitting and provide the optimal algorithm for spreading quantum information over any tree-topology network connecting multiple parties.

We evaluate the entanglement cost of spreading quantum information, using the following notations.
Given any tree $T=(V,E)$,
we let $\tilde{\Phi}_{D,e}^{+}$ for each $e=\{p\left(v_k\right),v_k\}\in E$ denote the reduced state for $\ket{\tilde{\Phi}_D^+}$ on the system $\bigotimes_{v\in D'_{v_k}}\tilde{\mathcal{H}}^{v}$ shared among $v_k$ itself and the descendants of $v_k$, that is,
\begin{equation}
    \label{eq:encoding_reduced_state}
    \tilde{\Phi}_{D,e}^{+}\coloneqq\tr_{R\overline{D'_{v_k}}}\ket{\tilde{\Phi}_D^+}\Bra{\tilde{\Phi}_D^+},
\end{equation}
where $\overline{D'_{v_k}}=V\setminus D'_{v_k}$ and $\tr_{R\overline{D'_{v_k}}}$ is the partial trace of the system $\mathcal{H}^R\otimes\bigotimes_{v\in\overline{D'_{v_k}}}\tilde{\mathcal{H}}^{v}$.
We provide an optimal algorithm for the state transformation defined as Equation~\eqref{eq:decoding_state_transformation} in Proposition~\ref{lem:encoding_state_transformation} equivalent to spreading quantum information,
which yields the following theorem.
\begin{theorem}
\label{thm:spreading}
    \textit{Entanglement cost of spreading quantum information over trees.}
    Given any tree $T=(V,E)$ and any isometry $U$,
    spreading quantum information over $T$ for $U$ is achievable if and only if, for each $e\in E$,
    \begin{equation}
        \label{eq:encoding_cost}
        \log_2 M_e \geqq \log_2\rank\tilde{\Phi}_{D,e}^{+},
    \end{equation}
    where $\tilde{\Phi}_{D,e}^{+}$ is defined as Equation~\eqref{eq:encoding_reduced_state}.
\end{theorem}

\begin{proof}
    \textit{If part}:
    Given any tree $T = (V,E)$ with an ascending labeling and any isometry $U$,
    we construct an algorithm for the state transformation defined as Equation~\eqref{eq:encoding_state_transformation} in Proposition~\ref{lem:encoding_state_transformation} by applying exact state splitting in Lemma~\ref{lem:split} sequentially starting from the root $v_1\in V$,
    and we prove that this algorithm achieves the equality in~\eqref{eq:encoding_cost} for each $e\in E$.
    In this algorithm, the root party $v_1$ first locally applies the given isometry $U$ to $\Ket{\Phi_D^+}\in\mathcal{H}^R\otimes\mathcal{H}$ on $\mathcal{H}$ to obtain $\ket{\tilde{\Phi}_D^+}$, where $\tilde{\mathcal{H}}$ is located at $v_1$ at this moment. Then, the parties perform the following sub-algorithm using the exact state splitting sequentially in order $v_1,v_2,\ldots,v_N$ to spread the state $\tilde{\mathcal{H}}$.
    After all the parties performing the sub-algorithm, spreading quantum information is achieved.

    We describe the sub-algorithm for each party $v_k\in V$.
    At the beginning of $v_k$'s sub-algorithm, we assume that the party $v_k$ holds the reduced state of $\ket{\tilde{\Phi}_D^+}$ on $\bigotimes_{v\in D'_{v_k}}\tilde{\mathcal{H}}^v$, that is, the system for the parties corresponding to $v_k$ itself and $v_k$'s descendants.
    Note that this assumption is satisfied because of an ascending labeling.
    If $v_k$ has no child, $v_k$'s sub-algorithm terminates.
    Otherwise, for each child $c\in C_{v_k}$,
    party $v_k$ performs the exact state splitting in Lemma~\ref{lem:split}, where $v_k$ and $c$ in the sub-algorithm are regarded as $A$ and $B$ in Lemma~\ref{lem:split}, and the subsystem $\bigotimes_{v\in D'_c}\tilde{\mathcal{H}}^v$, the other subsystems of party $v_k$, and all the rest of the system of the parties other than $v_k$ in the sub-algorithm are regarded as $\mathcal{H}^{A^\prime}$, $\mathcal{H}^A$, and $\mathcal{H}^{R^\prime}$ in Lemma~\ref{lem:split}, respectively.
    For each edge $e\in E$, Lemma~\ref{lem:split} shows that the exact state splitting in this sub-algorithm achieves the equality in~\eqref{eq:encoding_cost}, where $e=\left\{v_k,c\right\}$ in the above case.

    \textit{Only if part}:
    We use LOCC monotonicity of the Schmidt rank~\cite{L2} in the state transformation defined as Equation~\eqref{eq:encoding_state_transformation} in Proposition~\ref{lem:encoding_state_transformation}.
    Consider an arbitrary edge $e=\left\{p(v_k),v_k\right\}\in E$ where $v_k\neq v_1$.
    The Schmidt rank of the initial state $\Ket{\Phi_D^+}^{Rv_1}\otimes\bigotimes_{e\in E}\Ket{\Phi_{M_e}^+}^{e}$ between the parties in $D'_{v_k}$, that is, $v_k$ itself and $v_k$'s descendants, and the other parties in $\left\{R\right\}\cup V\setminus D'_{v_k}$ is $M_e$.
    After performing an LOCC map $\id^R\otimes\mathcal{S}$, the Schmidt rank of the final state $\ket{\tilde{\Phi}_D^+}^{Rv_1\cdots v_N}$ with respect to the same bipartition is $\rank\tilde{\Phi}_{D,e}^{+}$.
    Since the Schmidt rank is monotonically non-increasing under LOCC,
    we have
    $M_e\geqq\rank\tilde{\Phi}_{D,e}^{+}$.
    Therefore,
    we conclude
    $\log_2 M_e\geqq\log_2\rank\tilde{\Phi}_{D,e}^{+}$
    for each $e\in E$.
\end{proof}

\section{\label{sec:decoding}Entanglement cost of concentrating quantum information}
We derive an upper bound of entanglement cost of concentrating quantum information over any tree for any isometry, and show that the entanglement cost of concentrating quantum information is not larger than that of spreading quantum information, and can be strictly smaller.
To evaluate the entanglement cost, we generalize the two-party algorithm for exact state merging and provide an algorithm for concentrating quantum information over any tree-topology network connecting multiple parties.

We evaluate the entanglement cost of concentrating quantum information, using the following notations.
Given any tree $T=(V,E)$ and any isometry $U$,
our algorithm achieves the state transformation defined as Equation~\eqref{eq:decoding_state_transformation} in Proposition~\ref{lem:encoding_state_transformation} equivalent to spreading quantum information.
We write the initial state shared among $R,v_1,\ldots,v_N$ as
$\Ket{\Phi^N}\coloneqq\ket{\tilde{\Phi}_D^+}\in\mathcal{H}^R\otimes\tilde{\mathcal{H}}$,
and the states during the algorithm as a sequence
\begin{equation}
    \label{eq:sequence}
    \Ket{\Phi^N}\rightarrow\Ket{\Phi_{\boldsymbol{m}_{N-1}}^{N-1}}\rightarrow\cdots\rightarrow\Ket{\Phi_{\boldsymbol{m}_1}^1},
\end{equation}
where the subscript
$\boldsymbol{m}_k\coloneqq\left({m}^{v_N},\ldots,{m}^{v_{k+1}}\right)$
denotes a tuple representing measurement outcomes obtained during the algorithm and
\begin{equation}
    \begin{split}
        \label{eq:v_k}
        \Ket{\Phi_{\boldsymbol{m}_k}^k}\in\mathcal{H}^R\otimes\bigotimes_{v\in V_k}\tilde{\mathcal{H}}^{v},
        \quad V_k\coloneqq\left\{v_1,\ldots,v_k\right\},
    \end{split}
\end{equation}
for each $k\in\{1,\ldots,N-1\}$.
For any $\boldsymbol{m}_1$, the last state $\Ket{\Phi_{\boldsymbol{m}_1}^1}$ in sequence~\eqref{eq:sequence} is convertible into $\Ket{\Phi_D^+}\in\mathcal{H}^R\otimes\mathcal{H}$ by a local isometry transformation by party $v_1$, and recurrence relation to determine sequence~\eqref{eq:sequence} is given in the proof of the following theorem (by Equation~\eqref{eq:recursive_state}).
The following theorem uses the Koashi-Imoto decomposition of $\Ket{\Phi_{\boldsymbol{m}_k}^k}$ presented in Lemma~\ref{lem:koashi_imoto_decomposition_tripartite}, which yields
decompositions
\begin{equation}
   \begin{split}
    &\tilde{\mathcal{H}}^{v_k}=\bigoplus_{j=0}^{J_{\boldsymbol{m}_k}-1}\mathcal{H}_{\boldsymbol{m}_k}^{{\left(v_k\right)}_j^L}\otimes\mathcal{H}_{\boldsymbol{m}_k}^{{\left(v_k\right)}_j^R},\quad
    \bigotimes_{v\in V_{k-1}}\tilde{\mathcal{H}}^{v}=\bigoplus_{j=0}^{J_{\boldsymbol{m}_k}-1}\mathcal{H}_{\boldsymbol{m}_k}^{{\left(v_1\cdots v_{k-1}\right)}_j^L}\otimes\mathcal{H}_{\boldsymbol{m}_k}^{{\left(v_1\cdots v_{k-1}\right)}_j^R},\\
    &\Ket{\Phi_{\boldsymbol{m}_k}^k}=\bigoplus_{j=0}^{J_{\boldsymbol{m}_k}-1}\sqrt{p_{\boldsymbol{m}_k}\left(j\right)}\Ket{\omega_{\boldsymbol{m}_k,j}}\otimes\Ket{\phi_{\boldsymbol{m}_k,j}},
   \end{split}
\end{equation}
where $V_{k-1}$ is defined as Equation~\eqref{eq:v_k},
$p_{\boldsymbol{m}_k}\left(j\right)$ is a probability distribution, and,
for each $j\in\{0,\ldots,J_{\boldsymbol{m}_k}-1\}$,
$\Ket{\omega_{\boldsymbol{m}_k, j}}\in\mathcal{H}_{\boldsymbol{m}_k}^{{\left(v_k\right)}_j^L}\otimes\mathcal{H}_{\boldsymbol{m}_k}^{{\left(v_1\cdots v_{k-1}\right)}_j^L}$ and
$\Ket{\phi_{\boldsymbol{m}_k,j}}\in\mathcal{H}^R\otimes\mathcal{H}_{\boldsymbol{m}_k}^{{\left(v_k\right)}_j^R}\otimes\mathcal{H}_{\boldsymbol{m}_k}^{{\left(v_1\cdots v_{k-1}\right)}_j^R}$.
Also we let $\lambda_{\boldsymbol{m}_k,0}^{{\left(v_k\right)}_j^L}$ denote the largest eigenvalue of the reduced state of $\tr_{{\left(v_1\cdots v_{k-1}\right)}_j^L} \Ket{\omega_{\boldsymbol{m}_k,j}}\Bra{\omega_{\boldsymbol{m}_k,j}}$ of $\Ket{\omega_{\boldsymbol{m}_k,j}}$ on $\mathcal{H}_{\boldsymbol{m}_k}^{{\left(v_k\right)}_j^L}$.
\begin{theorem}
\label{thm:concentrating}
    \textit{Entanglement cost of concentrating quantum information.}
    Given any tree $T = (V,E)$ and any isometry $U$,
    concentrating quantum information over $T$ for $U$ is achievable if
    there exists an ascending labeling of the vertices satisfying, for each $e=\left\{p\left(v_k\right),v_k\right\}\in E$,
    \begin{equation}
        \label{eq:decoding_cost_upper}
        \log_2 M_e \geqq \max_{\boldsymbol{m}_k,j}\left\{\log_2\left\lceil\lambda_{\boldsymbol{m}_k,0}^{{\left(v_k\right)}_j^L}\dim\mathcal{H}_{\boldsymbol{m}_k}^{{\left(v_k\right)}_j^R}\right\rceil\right\},
    \end{equation}
    where $\lceil{}\cdots{}\rceil$ is the ceiling function.
\end{theorem}

\begin{proof}
    Given any tree $T = (V,E)$ with an ascending labeling of the vertices and any isometry $U$,
    we construct an algorithm for the state transformation defined as Equation~\eqref{eq:decoding_state_transformation} in Proposition~\ref{lem:encoding_state_transformation} achieving the equality in~\eqref{eq:decoding_cost_upper} for each $e\in E$.
    In the algorithm, the parties other than the root $v_1$ sequentially perform a sub-algorithm using exact state merging presented in Lemma~\ref{lem:merge}, where each of the parties $v_N,\ldots,v_2$ in this order is regarded as the sender $A$ in these sequential applications of the exact state merging.
    After all of these parties performing the sub-algorithm, the root party $v_1$ performs an isometry to obtain the state $\Ket{\Phi_D^+}$, which achieves concentrating quantum information.
    In the following, we describe the sub-algorithm and the isometry for the root party $v_1$.

    For any $v_k\in \left\{v_N,\ldots,v_2\right\}$, we describe the sub-algorithm for party $v_k$.
    We may write the state $\Ket{\Phi^{N}}=\Ket{\tilde{\Phi_D^+}}$ as $\Ket{\Phi_{\boldsymbol{m}_{N}}^{N}}$ for brevity.
    At the beginning of $v_k$'s sub-algorithm, we assume that the party $v_k$ has the reduced state on $\tilde{\mathcal{H}}^{v_k}$ of
    \begin{equation}
       \label{eq:assumption_merge}
       \Ket{\Phi_{\boldsymbol{m}_{k}}^{k}}\in\mathcal{H}^R\otimes\tilde{\mathcal{H}}^{v_k}\otimes\bigotimes_{m=1}^{k-1}\tilde{\mathcal{H}}^{v_m}.
    \end{equation}
    Based on the classical information $m^{v_N},\ldots,m^{v_{k+1}}$ of measurement outcomes sent from other parties by classical communication, the party $v_k$ calculates the measurement basis ${\left\{\Ket{m^{v_k}}\right\}}_{m^{v_k}}$ used in Equation~\eqref{eq:locc_merge} in Lemma~\ref{lem:merge} for the exact state merging of $\Ket{\Phi_{\boldsymbol{m}_{k}}^{k}}$ using the results presented in Reference~\cite{Y11},
    in which the systems $\mathcal{H}^R$, $\tilde{\mathcal{H}}^{v_k}$, and $\bigotimes_{m=1}^{k-1}\tilde{\mathcal{H}}^{v_m}$ are regarded as $\mathcal{H}^{R^\prime}$, $\mathcal{H}^A$, and $\mathcal{H}^B$ in Equation~\eqref{eq:locc_merge}, respectively.
    The party $v_k$ performs this measurement,
    and the states in the sequence~\eqref{eq:sequence} are recursively described as
    \begin{equation}
        \label{eq:recursive_state}
        \Ket{\Phi_{\boldsymbol{m}_{k-1}}^{k-1}}=\left(\mathbbm{1}\otimes\Bra{m^{v_k}}\right)\left(\Ket{\Phi_{\boldsymbol{m}_{k}}^{k}}\otimes\Ket{\Phi_{M_{e}}^+}^{e}\right),
    \end{equation}
    where $\mathbbm{1}$ is the identity operator on the system of the parties other than $v_k$, $\Ket{\Phi_{M_{e}}^+}^{e}$ with $e=\left\{p\left(v_k\right),v_k\right\}$ is the resource state shared between $v_k$ and $v_k$'s parent $p\left(v_k\right)$, and the system of party $p\left(v_k\right)$ for the resource state $\Ket{\Phi_{M_{e}}^+}^{e}$ on the right hand side is regarded on the left hand side as part of $\tilde{\mathcal{H}}^{p\left(v_k\right)}$ of the party $p\left(v_k\right)$.
    After this measurement, the party $v_k$ sends the measurement outcome $m^{v_k}$ to all the parties by classical communication, where the post-measurement state is represented by $\Ket{\Phi_{\boldsymbol{m}_{k-1}}^{k-1}}$.
    Note that the assumption~\eqref{eq:assumption_merge} is satisfied for the next party $v_{k-1}$ performing the sub-algorithm, that is,
    \begin{equation}
       \Ket{\Phi_{\boldsymbol{m}_{k-1}}^{k-1}}\in\mathcal{H}^R\otimes\tilde{\mathcal{H}}^{v_{k-1}}\otimes\bigotimes_{m=1}^{k-2}\tilde{\mathcal{H}}^{v_m},
    \end{equation}
    because of an ascending order of the vertices.
    For each edge $e=\left\{p\left(v_k\right),v_k\right\}\in E$, Lemma~\ref{lem:merge} shows that the exact state merging in this sub-algorithm achieves the equality in~\eqref{eq:decoding_cost_upper}.

    As for the party $v_1$, we derive an isometry $U_{\boldsymbol{m}_1}^{v_1}$ to obtain the state $\Ket{\Phi_D^+}$.
    After the parties $v_N,\ldots,v_2$ performing the above sub-algorithm, the shared state reduces to
    $\Ket{\Phi_{\boldsymbol{m}_{1}}^{1}}\in\mathcal{H}^R\otimes\tilde{\mathcal{H}}^{v_1}$.
    For each $v_k\in\left\{v_N,\ldots,v_2\right\}$,
    the isometry $U_{m^{v_k}}$ corresponding to $U_m^B$ in Equation~\eqref{eq:locc_merge} in Lemma~\ref{lem:merge} is used to recover the state $\Ket{\Phi_{\boldsymbol{m}_{k}}^{k}}$ from the post-measurement state $\Ket{\Phi_{\boldsymbol{m}_{k-1}}^{k-1}}$ corresponding to $\Bra{m^{v_k}}$, that is,
    $\Ket{\Phi_{\boldsymbol{m}_{k}}^{k}}=U_{m^{v_k}}\Ket{\Phi_{\boldsymbol{m}_{k-1}}^{k-1}}$.
    Repeating the above yields
    \begin{equation}
        \ket{\tilde{\Phi}_D^+}=\Ket{\Phi^N}=U_{m^{v_N}}\cdots U_{m^{v_2}}\Ket{\Phi_{\boldsymbol{m}_{1}}^{1}}.
    \end{equation}
    Consequently, the party $v_1$ obtains,
    for any $\boldsymbol{m}_1$,
    \begin{equation}
        \Ket{\Phi_D^+}=U_{\boldsymbol{m}_1}^{v_1}\Ket{\Phi_{\boldsymbol{m}_{1}}^{1}},\quad 
        U_{\boldsymbol{m}_1}^{v_1}\coloneqq U^\dag U_{m^{v_N}}\cdots U_{m^{v_2}}.
    \end{equation}
    Note that it may not be possible for the parties $v_N,\ldots,v_2$ to locally perform $U_{m^{v_N}},\ldots,U_{m^{v_2}}$ during the sub-algorithm, since these isometries can be nonlocal.
\end{proof}

We show that the entanglement cost of concentrating quantum information is not larger than that of spreading quantum information, as presented in the following theorem.
Note that the former can be strictly smaller than the latter, as we demonstrate in Application~\ref{ex:1} and~\ref{ex:2} in the next section.
\begin{theorem}
    \textit{Comparison of entanglement cost between spreading and concentrating quantum information.}
    Given any tree $T=(V,E)$ with any ascending labeling and any isometry $U$,
    \begin{equation}
        \max_{\boldsymbol{m}_k,j}\left\{\log_2\left\lceil\lambda_{\boldsymbol{m}_k,0}^{{\left(v_k\right)}_j^L}\dim\mathcal{H}_{\boldsymbol{m}_k}^{{\left(v_k\right)}_j^R}\right\rceil\right\}\leqq\log_2\rank\tilde{\Phi}_{D,e}^{+}
    \end{equation}
    where the notations are the same as those in Theorems~\ref{thm:spreading} and~\ref{thm:concentrating}.
\end{theorem}

\begin{proof}
   We use LOCC monotonicity of the Schmidt rank~\cite{L2} in the state transformation defined as Equation~\eqref{eq:decoding_state_transformation} in Proposition~\ref{lem:encoding_state_transformation}, and properties of the Koashi-Imoto decomposition.
    We regard the given tree $T=(V,E)$ as the rooted tree with its root $v_1$, and we consider an arbitrary edge $e=\left\{p(v_k),v_k\right\}\in E$ where $v_k\neq v_1$.
    The Schmidt rank of the initial state
    $\ket{\tilde{\Phi}_D^+}^{Rv_1\cdots v_N}\otimes\bigotimes_{e\in E}\Ket{\Phi_{M_e}^+}^{e}$
    between the parties in $D'_{v_k}$ and the other parties in $\{R\}\cup V\setminus D'_{v_k}$ is $M_e\rank\tilde{\Phi}_{D,e}^{+}$.
    After the parties $v_N,\ldots,v_{k-1}$ performing the above sub-algorithms, which is an LOCC map, the state reduces to $\Ket{\Phi_{\boldsymbol{m}_k}^k}\otimes\bigotimes_{e\in E_k}\Ket{\Phi_{M_e}^+}^{e}$,
    where $\Ket{\Phi_{\boldsymbol{m}_k}^k}$ is defined as Equation~\eqref{eq:v_k} and
    $E_k\coloneqq\left\{\left\{p\left(v_2\right),v_2\right\},\ldots,\left\{p\left(v_k\right),v_k\right\}\right\}$.
    The Schmidt rank of $\Ket{\Phi_{\boldsymbol{m}_k}^k}\otimes\bigotimes_{e\in E_k}\Ket{\Phi_{M_e}^+}^{e}$ with respect to the same bipartition of the parties as the above is $M_e \rank{\left(\Phi_{\boldsymbol{m}_k}^k\right)}^{v_k}$, where ${\left(\Phi_{\boldsymbol{m}_k}^k\right)}^{v_k}$ denotes the reduced state of the system $\tilde{\mathcal{H}}^{v_k}$ for the state $\Ket{\Phi_{\boldsymbol{m}_k}^k}$.
    Since the Schmidt rank is monotonically nonincreasing under LOCC,
    it holds that
    \begin{equation}
        M_e\rank\tilde{\Phi}_{D,e}^{+}\geqq M_e\rank{\left(\Phi_{\boldsymbol{m}_k}^k\right)}^{v_k}.
    \end{equation}
    By construction of the Koashi-Imoto decomposition, it holds that for any ${\boldsymbol{m}_k}$ and $j$,
    \begin{equation}
        \rank{\left(\Phi_{\boldsymbol{m}_k}^k\right)}^{v_k}\geqq\dim\mathcal{H}_{\boldsymbol{m}_k}^{{\left(v_k\right)}_j^R}.
    \end{equation}
    Since $\lambda_{\boldsymbol{m}_k,0}^{{\left(v_k\right)}_j^L}\leqq 1$,
    we obtain
    \begin{equation}
        \dim\mathcal{H}_{\boldsymbol{m}_k}^{{\left(v_k\right)}_j^R}\geqq\left\lceil\lambda_{\boldsymbol{m}_k,0}^{{\left(v_k\right)}_j^L}\dim\mathcal{H}_{\boldsymbol{m}_k}^{{\left(v_k\right)}_j^R}\right\rceil.
    \end{equation}
    Thus, for any ${\boldsymbol{m}_k}$ and $j$, we have
    \begin{equation}
        \log_2\rank\tilde{\Phi}_{D,e}^{+}\geqq\log_2\left\lceil\lambda_{\boldsymbol{m}_k,0}^{{\left(v_k\right)}_j^L}\dim\mathcal{H}_{\boldsymbol{m}_k}^{{\left(v_k\right)}_j^R}\right\rceil.
    \end{equation}
    Therefore, we obtain
    \begin{equation}
        \max_{\boldsymbol{m}_k,j}\left\{\log_2\left\lceil\lambda_{\boldsymbol{m}_k,0}^{{\left(v_k\right)}_j^L}\dim\mathcal{H}_{\boldsymbol{m}_k}^{{\left(v_k\right)}_j^R}\right\rceil\right\}\leqq\log_2\rank\tilde{\Phi}_{D,e}^{+}
    \end{equation}
    for each $e=\left\{p\left(v_k\right),v_k\right\}\in E$.
\end{proof}

\section{\label{sec:example}Applications}
We provide applications of our algorithms for spreading and concentrating quantum information.
Define
\begin{equation}
   \Ket{\pm}\coloneqq\frac{1}{\sqrt{2}}\left(\Ket{0}\pm\Ket{1}\right).
\end{equation}
In the following, we may omit $\otimes$ if obvious.

\begin{application}
\label{ex:1}
\textit{Application to one-shot distributed source compression for arbitrarily-small-dimensional systems.}
When applied to a star-topology tree, such as
\begin{equation}
    \label{eq:star}
    \begin{split}
        T=(V,E),\quad
        V=\left\{v_1,v_2,v_3,v_4\right\},\quad
        E=\left\{e_1=\left\{v_1,v_2\right\},e_2=\left\{v_1,v_3\right\},e_2=\left\{v_1,v_4\right\}\right\},
    \end{split}
\end{equation}
our algorithm for concentrating quantum information can be regarded as an algorithm for one-shot distributed source compression.\cite{D8,D9,A8}
Although our algorithms achieve transformations between $\Ket{\Phi_D^+}$ and $\ket{\tilde{\Phi}_D^+}$, that is, maximally entangled states between $R$ and the others, it is straightforward to prove that our algorithms also work for any pure state shared among the parties $R, v_1,\ldots,v_N$, which is proven for two parties in Reference~\cite{Y11}, and the same argument also applies to more than two parties.
We remark that our algorithm for concentrating quantum information is applicable to arbitrarily-small-dimensional systems as well as achieving zero error, while the existing algorithms for the one-shot distributed source compression~\cite{D8,D9,A8} are inefficient for small- and intermediate-scale states and cannot avoid nonzero approximation error, similarly to the case of $N=2$.~\cite{Y11}

For the network defined as Equation~\eqref{eq:star} and an isometry mapping the basis states as
\begin{equation}
    \Ket{0}\leftrightarrow\Ket{0}^{v_1}\Ket{0}^{v_2}\Ket{0}^{v_3}\Ket{0}^{v_4},\quad
    \Ket{1}\leftrightarrow\Ket{+}^{v_1}\Ket{1}^{v_2}\Ket{1}^{v_3}\Ket{1}^{v_4},
\end{equation}
Theorem~\ref{thm:spreading} yields the entanglement cost of spreading quantum information
\begin{equation}
    \begin{split}
        \log_2 M_{e_1}=1,\quad
        \log_2 M_{e_2}=1,\quad
        \log_2 M_{e_3}=1,
    \end{split}
\end{equation}
and Theorem~\ref{thm:concentrating} yields an algorithm for concentrating quantum information achieving
\begin{equation}
    \begin{split}
        \log_2 M_{e_1}=1,\quad
        \log_2 M_{e_2}=0\neq 1,\quad
        \log_2 M_{e_3}=0\neq 1.
    \end{split}
\end{equation}
In concentrating quantum information, the states in sequence~\eqref{eq:sequence} are calculated as
\begin{align}
   \label{eq:1}
   &\Ket{\Phi^{4}}=\ket{\tilde{\Phi}_D^+}=\frac{1}{\sqrt{2}}\Ket{0}^R\Ket{0}^{{\left(v_4\right)}_0^R}\Ket{000}^{{\left(v_1 v_2 v_3\right)}_0^R}\oplus\frac{1}{\sqrt{2}}\Ket{1}^R\Ket{1}^{{\left(v_4\right)}_1^R}\Ket{+11}^{{\left(v_1 v_2 v_3\right)}_1^R}\\
   \label{eq:2}
   &\xrightarrow{\text{Measurement in }\left\{\Ket{\pm}^{v_4}\right\}}
   \Ket{\Phi_{\left(\Ket{\pm}^{v_4}\right)}^{3}}=\frac{1}{\sqrt{2}}\Ket{0}^R\Ket{0}^{{\left(v_3\right)}_0^R}\Ket{00}^{{\left(v_1 v_2\right)}_0^R}\oplus\left(\pm\frac{1}{\sqrt{2}}\Ket{1}^R\Ket{1}^{{\left(v_3\right)}_1^R}\Ket{+1}^{{\left(v_1 v_2\right)}_1^R}\right)\\
   \label{eq:3}
   &\xrightarrow{\text{Measurement in }\left\{\Ket{\pm}^{v_3}\right\}}
   \begin{cases}
      &\Ket{\Phi_{\left(\Ket{\pm}^{v_4},\Ket{\pm}^{v_3}\right)}^{2}}=\frac{1}{\sqrt{2}}\Ket{0}^R\Ket{0}^{{\left(v_2\right)}_0^R}\Ket{0}^{{\left(v_1\right)}_0^R}+\frac{1}{\sqrt{2}}\Ket{1}^R\Ket{1}^{{\left(v_2\right)}_0^R}\Ket{+}^{{\left(v_1\right)}_0^R}\\
      &\Ket{\Phi_{\left(\Ket{\pm}^{v_4},\Ket{\mp}^{v_3}\right)}^{2}}=\frac{1}{\sqrt{2}}\Ket{0}^R\Ket{0}^{{\left(v_2\right)}_0^R}\Ket{0}^{{\left(v_1\right)}_0^R}-\frac{1}{\sqrt{2}}\Ket{1}^R\Ket{1}^{{\left(v_2\right)}_0^R}\Ket{+}^{{\left(v_1\right)}_0^R},\\
   \end{cases}
\end{align}
where the right-hand sides of Equations~\eqref{eq:1},~\eqref{eq:2}, and~\eqref{eq:3} shows the Koashi-Imoto decomposition of the state for each step in the sequence~\eqref{eq:sequence},
and the final state shared between $R$ and $v_1$ is obtained by transferring $v_2$'s one-qubit state by quantum teleportation from $v_2$ to $v_1$, which requires $\log_2 M_{e_1}=1$.
The difference in the resource requirements for concentrating quantum information between the edges $e_1$ and $e_2,e_3$ arises because of the difference between the direct-sum structure of the Koashi-Imoto decomposition of the states $\Ket{\Phi_{\left(\Ket{\pm}^{v_4},\Ket{\pm}^{v_3}\right)}^{2}},\Ket{\Phi_{\left(\Ket{\pm}^{v_4},\Ket{\mp}^{v_3}\right)}^{2}}$ after the parties $v_4,v_3$ performing the exact state merging and that of the states $\Ket{\Phi^{4}},\Ket{\Phi_{\left(\Ket{\pm}^{v_4}\right)}^{3}}$ before.

By contrast, if we swap the labeling of the parties $v_2$ and $v_3$, the tree $T$ changes to
\begin{equation}
   \begin{split}
      &T'=(V',E),\quad
      V'=\left\{v'_1=v_1, v'_2=v_3, v'_3=v_2, v'_4=v_4\right\},\\
      &E=\left\{e_1=\left\{v'_1,v'_3\right\}=\left\{v_1,v_2\right\},e_2=\left\{v'_1,v'_2\right\}=\left\{v_1,v_3\right\},e_3=\left\{v'_1,v'_4\right\}=\left\{v_1,v_4\right\}\right\},
   \end{split}
\end{equation}
and our algorithm for concentrating quantum information applied to this tree $T'$ achieves
\begin{equation}
    \begin{split}
        \log_2 M_{e_1}=0,\quad
        \log_2 M_{e_2}=1,\quad
        \log_2 M_{e_3}=0.
    \end{split}
\end{equation}

This example implies that the entanglement cost of concentrating quantum information for each edge of a graph may be affected by the labeling of the vertices, that is, the order of sequential applications of exact state merging.
In this case, to obtain the entanglement cost, we need to calculate the Koashi-Imoto-decomposition of the state for each step of the sequence~\eqref{eq:sequence} in the algorithm, by recursively applying Equation~\eqref{eq:recursive_state}.
\end{application}

\begin{application}
\label{ex:2}
\textit{Application to LOCC-assisted decoding in quantum secret sharing.}
Similarly to our algorithm for concentrating quantum information, Reference~\cite{G4} proposes schemes of quantum secret sharing and an algorithm for decoding shared secret of quantum information, in which the parties collaboratively perform LOCC to reduce total quantum communication required for the decoding.
While the algorithm in Reference~\cite{G4} works for a particular class of quantum codes such as the seven-qubit code,
our algorithms are applicable to any encoding and decoding in addition to this particular class.
For example, a different scheme of quantum secret sharing from those considered in Reference~\cite{G4} can be obtained from the five-qubit code,\cite{C,G2} which maps the basis states as
\begin{equation}
   \begin{split}
    \Ket{0}\leftrightarrow\frac{1}{4}(&\Ket{00000} + \Ket{11000} + \Ket{01100} + \Ket{00110}
    +\Ket{00011}+\Ket{10001}-\Ket{10100}-\Ket{01010}\\
    -&\Ket{00101}-\Ket{10010}-\Ket{01001}-\Ket{11110}
    -\Ket{01111}-\Ket{10111}-\Ket{11011}-\Ket{11101}),\\
    \Ket{1}\leftrightarrow\frac{1}{4}(&\Ket{11111} + \Ket{00111} + \Ket{10011} + \Ket{11001}
    +\Ket{11100} + \Ket{01110}-\Ket{01011}-\Ket{10101}\\
    -&\Ket{11010}-\Ket{01101}-\Ket{10110}-\Ket{00001}
    -\Ket{10000}-\Ket{01000}-\Ket{00100}-\Ket{00010}),
   \end{split}
\end{equation}
where each qubit on the right hand sides belongs to each of the parties $v_1,\ldots,v_5$.
For this isometry and a line-topology tree
\begin{equation}
    \begin{split}
        T=(V,E),\quad
        V=\left\{v_k:k=1,\ldots,N\right\},\quad
        E=\left\{e_k=\left\{v_k,v_{k+1}\right\}:k=1,\ldots,N-1\right\},
    \end{split}
\end{equation}
where $N=5$,
Theorem~\ref{thm:spreading} yields the entanglement cost of spreading quantum information
\begin{equation}
    \begin{split}
        \log_2 M_{e_1}=2,\quad
        \log_2 M_{e_2}=3,\quad
        \log_2 M_{e_3}=2,\quad
        \log_2 M_{e_4}=1,
    \end{split}
\end{equation}
and Theorem~\ref{thm:concentrating} yields an algorithm for concentrating quantum information achieving
\begin{equation}
    \label{eq:ex3}
    \log_2 M_{e_1}=0,\quad
    \log_2 M_{e_2}=0,\quad
    \log_2 M_{e_3}=0,\quad
    \log_2 M_{e_4}=0.
\end{equation}
In concentrating quantum information, the states in sequence~\eqref{eq:sequence} are calculated as
\begin{equation}
    \begin{split}
       &\Ket{\Phi^{5}}=\ket{\tilde{\Phi}_D^+}\propto\Ket{+}^R\Ket{+}^{{\left(v_5\right)}_0^R}{\left(\Ket{0000}^{{\left(v_1 v_2 v_3 v_4\right)}_0^R}+\cdots\right)}\oplus\Ket{-}^R\Ket{-}^{{\left(v_5\right)}_1^R}{\left(\Ket{0000}^{{\left(v_1 v_2 v_3 v_4\right)}_1^R}+\cdots\right)}\\
       &\downarrow{\text{Measurement in }\left\{\Ket{0}^{v_5},\Ket{1}^{v_5}\right\}}\\
       &\Ket{\Phi_{\left(\Ket{0}^{v_5}\right)}^{4}}\\
       &=\frac{1}{4}\left[\Ket{0}^{R}\left(\Ket{0000}^{v_1 v_2 v_3 v_4} + \Ket{1100} + \Ket{0110} + \Ket{0011}
       -\Ket{1010}-\Ket{0101}-\Ket{1001}-\Ket{1111}\right)\right.\\
       &\quad\left.+\Ket{1}^R\left(\Ket{1110}^{v_1 v_2 v_3 v_4} + \Ket{0111}-\Ket{1101}-\Ket{1011}
       -\Ket{1000}-\Ket{0100}-\Ket{0010}-\Ket{0001}\right)\right]\\
       &\propto\Ket{+}^R\Ket{+}^{{\left(v_4\right)}_0^R}\left(\Ket{000}^{{\left(v_1 v_2 v_3\right)}_0^R}+\cdots\right)\oplus\Ket{-}^R\Ket{-}^{{\left(v_4\right)}_0^R}\left(\Ket{000}^{{\left(v_1 v_2 v_3\right)}_0^R}+\cdots\right)\\
       &\downarrow{\text{Measurement in }\left\{\Ket{0}^{v_4},\Ket{1}^{v_4}\right\}}\\
       &\Ket{\Phi_{\left(\Ket{0}^{v_5},\Ket{0}^{v_4}\right)}^{3}}\\
       &=\frac{1}{2\sqrt{2}}\left[\Ket{0}^R\left(\Ket{000}^{v_1 v_2 v_3} + \Ket{110} + \Ket{011} - \Ket{101}\right) +\Ket{1}^R\left(\Ket{111}^{v_1 v_2 v_3} - \Ket{100}-\Ket{010}-\Ket{001}\right)\right]\\
       &\propto\Ket{+}^R\Ket{+}^{{\left(v_3\right)}_0^R}\left(\Ket{00}^{{\left(v_1 v_2\right)}_0^R}+\cdots\right)\oplus\Ket{-}^R\Ket{-}^{{\left(v_3\right)}_0^R}\left(\Ket{00}^{{\left(v_1 v_2\right)}_0^R}+\cdots\right)\\
       &\downarrow{\text{Measurement in }\left\{\Ket{0}^{v_3},\Ket{1}^{v_3}\right\}}\\
       &\Ket{\Phi_{\left(\Ket{0}^{v_5},\Ket{0}^{v_4},\Ket{0}^{v_3}\right)}^{2}}\\
       &=\frac{1}{2}\left[\Ket{0}^R\left(\Ket{00}^{v_1 v_2}+\Ket{11}\right)-\Ket{1}^R\left(\Ket{01}^{v_1 v_2}+\Ket{10}\right)\right]\\
       &\propto\Ket{-}^R\Ket{+}^{{\left(v_2\right)}_0^R}\Ket{+}^{{\left(v_1\right)}_0^R}\oplus\Ket{+}^R\Ket{-}^{{\left(v_2\right)}_1^R}\Ket{-}^{{\left(v_1\right)}_1^R},\\
       &\downarrow{\text{Measurement in }\left\{\Ket{0}^{v_2},\Ket{1}^{v_2}\right\}}\\
       &\Ket{\Phi_{\left(\Ket{0}^{v_5},\Ket{0}^{v_4},\Ket{0}^{v_3},\Ket{0}^{v_2}\right)}^{1}}=\frac{1}{\sqrt{2}}\left(\Ket{0}^{R}\Ket{0}^{v_1}-\Ket{1}^R\Ket{1}^{v_1}\right)\\
       &\downarrow{\text{Local isometry by }v_1}\\
       &\frac{1}{\sqrt{2}}\left(\Ket{0}^{R}\Ket{0}^{v_1}+\Ket{1}^R\Ket{1}^{v_1}\right)
    \end{split}
\end{equation}
where the Koashi-Imoto decomposition of the state for each step in the sequence~\eqref{eq:sequence} is shown after $\propto$ for the above states,
and we only show the sequence of states for the measurement outcomes corresponding to $\Ket{0}$'s, while those corresponding to other outcomes can be calculated in the same way.
Equation~\eqref{eq:ex3} shows that the five-qubit code can be decoded only by LOCC, \textit{i.e.,} without quantum communication.
Note that, if our algorithms are applied to quantum secret sharing, some subsets of the parties may extract partial knowledge about the shared secret of quantum information during the algorithms while this is the same situation as the existing algorithm in Reference~\cite{G4}.
\end{application}

\section{\label{sec:conclusion}Conclusion}
We quantitatively characterized nonlocal properties of multipartite quantum transformations for encoding and decoding quantum information in a multipartite system in terms of the entanglement cost.
For any tree-topology network connecting spatially separated parties $v_1,\ldots,v_N$,
we evaluated the entanglement costs of performing an isometry $U:\mathcal{H}\to\bigotimes_{k=1}^{N}\tilde{\mathcal{H}}^{v_k}$ representing encoding and the inverse $U^\dag:\bigotimes_{k=1}^{N}\tilde{\mathcal{H}}^{v_k}\to\mathcal{H}$ representing decoding, where the system $\mathcal{H}$ for logical states is located at one of the parties and each subsystem $\tilde{\mathcal{H}}^{v_k}$ for physical states is located at each party $v_k$.
Regarding the encoding, our algorithm for spreading quantum information is proven to achieve the optimal entanglement cost.
As for the decoding, our algorithm for concentrating quantum information can reduce the entanglement cost compared to that of spreading quantum information.
Hence, while $U$ and $U^\dag$ are inverse of each other, we derived bounds for quantitatively differentiating nonlocal properties of $U$ for encoding and $U^\dag$ for decoding in terms of entanglement cost.
We also demonstrated applications of our algorithms to multiparty tasks such as one-shot distributed source compression~\cite{D8,D9,A8} and LOCC-assisted decoding in quantum secret sharing~\cite{G4}.
The concept of encoding and the decoding represented by isometries has essential roles not only in quantum information science,
and we leave further investigation of applications within and beyond quantum information science for future works.

\section*{Acknowledgments}
This work was supported by Grant-in-Aid for JSPS Research Fellow and JSPS KAKENHI Grant Numbers 26330006, 15H01677, 16H01050, 17H01694, 18H04286, and 18J10192.


\begin{thebibliography}{10}
    \bibitem{G}D.\ Gottesman, arXiv:0904.2557.
    \bibitem{D}S.\ J.\ Devitt, W.\ J.\ Munro,  K.\ Nemoto, \textit{Rep.\ Prog.\ Phys.} \textbf{2013}, \textit{76}, 076001.
    \bibitem{T2}B.\ M.\ Terhal, \textit{Rev.\ Mod.\ Phys.} \textbf{2015}, \textit{87}, 307.
    \bibitem{B}B.\ J.\ Brown, D.\ Loss, J.\ K.\ Pachos, C.\ N.\ Self, J.\ R.\ Wootton, \textit{Rev.\ Mod.\ Phys.} \textbf{2016}, \textit{88}, 045005.
    \bibitem{K}A.\ Y.\ Kitaev, \textit{Ann.\ Phys.} \textbf{2003}, \textit{303}, 2.
    \bibitem{K2}A.\ Y.\ Kitaev, \textit{Ann.\ Phys.} \textbf{2006}, \textit{321}, 2.
    \bibitem{A}A.\ Almheiri, X.\ Dong,  D.\ Harlow, \textit{J.\ High Energy Phys.} \textbf{2015}, \textit{04}, 163.
    \bibitem{P}F.\ Pastawski, B.\ Yoshida, D.\ Harlow,  J.\ Preskill, \textit{J.\ High Energy Phys.} \textbf{2015}, \textit{06}, 149.
    \bibitem{F}F.\ G.\ S.\ L. Brand\~{a}o, E.\ Crosson, M.\ B.\ \c{S}ahino\u{g}lu,  J.\ Bowen, arXiv:1710.04631.
    \bibitem{B7}M.\ Hillery, V.\ Bu\u{z}ek,  A.\ Berthiaume, \textit{Phys.\ Rev.\ A} \textbf{1999}, \textit{59}, 1829.
    \bibitem{C}R.\ Cleve, D.\ Gottesman,  H.-K.\ Lo, \textit{Phys.\ Rev.\ Lett.} \textbf{1999}, \textit{83}, 648.
    \bibitem{G2}D.\ Gottesman, \textit{Phys.\ Rev.\ A} \textbf{2000}, \textit{61}, 042311.
    \bibitem{C8}E.\ Chitambar, G.\ Gour, arXiv:1806.06107.
    \bibitem{H2}R.\ Horodecki, P.\ Horodecki, M.\ Horodecki,  K.\ Horodecki, \textit{Rev.\ Mod.\ Phys.} \textbf{2009}, \textit{81}, 865.
    \bibitem{P2}M.\ B.\ Plenio, S.\ Virmani, \textit{Quantum Inf.\ Comput.} \textbf{2007}, \textit{7}, 1.
    \bibitem{L}E.\ H.\ Lieb, J.\ Yngvason, \textit{Phys.\ Rep.} \textbf{1999}, \textit{310}, 1.
    \bibitem{B5}C.\ H.\ Bennett, G.\ Brassard, C.\ Crepeau, R.\ Jozsa, A.\ Peres,  W.\ K.\ Wootters, \textit{Phys.\ Rev.\ Lett.} \textbf{1993}, \textit{70}, 1895.
    \bibitem{B2}C.\ H.\ Bennett, D.\ P.\ DiVincenzo, J.\ A.\ Smolin,  W.\ K.\ Wootters, \textit{Phys.\ Rev.\ A} \textbf{1996}, \textit{54}, 3824.
    \bibitem{H}P.\ M.\ Hayden, M.\ Horodecki,  B.\ M.\ Terhal, \textit{J.\ Phys.\ A} \textbf{2001}, \textit{34}, 6891.
    \bibitem{T}B.\ M.\ Terhal, P.\ Horodecki, \textit{Phys.\ Rev.\ A} \textbf{2000}, \textit{61}, 040301.
    \bibitem{Z}X.\ Zhou, D.\ W.\ Leung,  I.\ L.\ Chuang, \textit{Phys. Rev. A} \textbf{2000}, \textit{62}, 052316.
    \bibitem{E}J.\ Eisert, K.\ Jacobs, P.\ Papadopoulos,  M.\ B.\ Plenio, \textit{Phys. Rev. A} \textbf{2000}, \textit{62} 052317.
    \bibitem{C3}A.\ Chefles, C.\ R.\ Gilson,  S.\ M.\ Barnett, \textit{Phys.\ Rev.\ A} \textbf{2001}, \textit{63}, 032314.
    \bibitem{C2}D.\ Collins, N.\ Linden,  S.\ Popescu, \textit{Phys.\ Rev.\ A} \textbf{2001}, \textit{64}, 032302.
    \bibitem{N}M.\ A.\ Nielsen, C.\ M.\ Dawson, J,\ L.\ Dodd, A.\ Gilchrist, D.\ Mortimer, T.\ J.\ Osborne, M.\ J.\ Bremner, A.\ W.\ Harrow,  A.\ Hines, \textit{Phys.\ Rev.\ A} \textbf{2003}, \textit{67}, 052301.
    \bibitem{Y5}C.-P.\ Yang, \textit{Phys.\ Lett.\ A} \textbf{2008}, \textit{372}, 9, 25, 1380.
    \bibitem{C4}S.\ M.\ Cohen, \textit{Phys.\ Rev.\ A} \textbf{2010}, \textit{81}, 062316.
    \bibitem{Y}L.\ Yu, R.\ B.\ Griffiths,  S.\ M.\ Cohen, \textit{Phys.\ Rev.\ A} \textbf{2010}, \textit{81}, 062315.
    \bibitem{S}D.\ Stahlke, R.\ B.\ Griffiths, \textit{Phys.\ Rev.\ A} \textbf{2011}, \textit{84}, 032316.
    \bibitem{S2}A.\ Soeda, P.\ S.\ Turner,  M.\ Murao, \textit{Phys.\ Rev.\ Lett.} \textbf{2011}, \textit{107}, 180501.
    \bibitem{Y4}L.\ Chen, L.\ Yu, \textit{Phys.\ Rev.\ A} \textbf{2014}, \textit{89}, 062326.
    \bibitem{Y2}L.\ Chen, L.\ Yu, \textit{Ann.\ Phys.} \textbf{2014}, \textit{351}, 682.
    \bibitem{S3}D.\ Saha, S.\ Nandan,  P.\ K.\ Panigrahi, J.\ \textit{Quantum Inf.\ Sci.} \textbf{2014}, \textit{4}, 2, 46117.
    \bibitem{X}L.-P.\ Xue, M.\ Jiang, \textit{2015 34th Chinese Control Conference (CCC)} \textbf{2015}, \textit{6636}.
    \bibitem{C5}L.\ Chen, L.\ Yu, \textit{Phys.\ Rev.\ A} \textbf{2016}, \textit{93}, 042331.
    \bibitem{V}N.\ Vyas, D.\ Saha,  P.\ K. Panigrahi, \textit{Quantum Inf.\ Processing} \textbf{2016}, \textit{15}, 3855.
    \bibitem{Y3}L.\ Yu, K.\ Nemoto, \textit{Phys.\ Rev.\ A} \textbf{2016}, \textit{94}, 022320.
    \bibitem{W}E.\ Wakakuwa, A.\ Soeda,  M.\ Murao, \textit{IEEE Trans.\ Inf.\ Theory} \textbf{2017}, \textit{63}, 5372.
    \bibitem{W2}E.\ Wakakuwa, A.\ Soeda,  M.\ Murao, arXiv:1608.07461.
    \bibitem{W5}E.\ Wakakuwa, A.\ Soeda,  M.\ Murao, arXiv:1810.08447.
    \bibitem{J}R.\ Jozsa, M.\ Koashi, N.\ Linden, S.\ Popescu, S.\ Presnell, D.\ Shepherd,  A.\ Winter, \textit{Quantum Inf.\ Comput.} \textbf{2003}, \textit{3}, 5, 405.
    \bibitem{B3}S.\ Bandyopadhyay, G.\ Brassard, S.\ Kimmel,  W.\ K.\ Wootters, \textit{Phys.\ Rev.\ A} \textbf{2009}, \textit{80}, 012313.
    \bibitem{B4}S.\ Bandyopadhyay, R.\ Rahaman,  W.\ K.\ Wootters, \textit{J.\ Phys.\ A} \textbf{2010}, \textit{43}, 45.
    \bibitem{Y6}H.\ Yamasaki, A.\ Soeda,  M.\ Murao, \textit{Phys.\ Rev.\ A} \textbf{2017}, \textit{96}, 032330.
    \bibitem{G3}E.\ F.\ Galv\~{a}o, L.\ Hardy, \textit{Phys.\ Rev.\ A} \textbf{2000}, \textit{62}, 012309.
    \bibitem{Y7}S.\ Yang, H.\ Jeong, \textit{Phys.\ Rev.\ A} \textbf{2015}, \textit{92}, 022322.
    \bibitem{E2}J.\ Eisert, D.\ Gross, in \textit{Lectures on Quantum Information}, (Eds: D.\ Bru\ss, G.\ Leuchs), Wiley, Weinheim \textbf{2007}, Ch.\ 13.
    \bibitem{W3}M.\ Walter, D.\ Gross,  J.\ Eisert, arXiv:1612.02437.
    \bibitem{B8}I.\ Bengtsson, K.\ \.{Z}yczkowski, \textit{Geometry of Quantum States: An Introduction to Quantum Entanglement}, Cambridge University Press, New York \textbf{2017}, Ch.\ 17.
    \bibitem{K6}H.\ J.\ Kimble, \textit{Nature} \textbf{2008}, \textit{453}, 1023.
    \bibitem{B6}J.\ A.\ Bondy, U.\ S.\ R.\ Murty, \textit{Graph Theory}, Springer, London \textbf{2008}.
    \bibitem{Y12}H.\ Yamasaki, A.\ Pirker, M.\ Murao, W.\ D\"{u}r, B.\ Kraus, \textit{Phys.\ Rev.\ A} \textbf{2018}, \textit{98}, 052313.
    \bibitem{F3}B.\ Fortescue, G.\ Gour, \textit{IEEE Trans.\ Inf.\ Theory} \textbf{2012}, \textit{58}, 6659.
    \bibitem{S5}K.\ Senthoor, P.\ K.\ Sarvepalli, arXiv:1801.09500.
    \bibitem{Y11}H.\ Yamasaki, M.\ Murao, arXiv:1806.07875.
    \bibitem{C6}R.\ H.\ Choi, B.\ Fortescue, G.\ Gour,  B.\ C.\ Sanders, \textit{Phys.\ Rev.\ A} \textbf{2013}, \textit{87}, 032319.
    \bibitem{R}R.\ Rahaman,  M.\ G.\ Parker, \textit{Phys.\ Rev.\ A} \textbf{2015}, \textit{91}, 022330.
    \bibitem{Y10}Y.-H.\ Yang, F.\ Gao, X.\ Wu, S.-J.\ Qin, H.-J.\ Zuo,  Q.-Y.\ Wen, \textit{Sci.\ Rep.} \textbf{2015}, \textit{5}, 16967.
    \bibitem{W6}J.\ Wang, L.\ Li, H.\ Peng,  Y.\ Yang, \textit{Phys.\ Rev.\ A} \textbf{2017}, \textit{95}, 022320.
    \bibitem{B11}C.-M.\ Bai, Z.-H. Li, C.-J.\ Liu, Y.-M.\ Li, \textit{Quantum Inf.\ Processing} \textbf{2017}, \textit{16}, 304.
    \bibitem{L3}C.-J.\ Liu, Z.-H.\ Li, C.-M.\ Bai,  M.-M.\ Si, \textit{Int.\ J.\ Theor.\ Phys.} \textbf{2018}, \textit{57} 428.
    \bibitem{D8}N.\ Dutil, P.\ Hayden, arXiv:1011.1974.
    \bibitem{D9}N.\ Dutil, \textit{PhD Thesis}, McGill University, May, \textbf{2011}.
    \bibitem{A8}A.\ Anshu, R.\ Jain, N.\ A.\ Warsi, \textit{IEEE Trans.\ Inf.\ Theory} \textbf{2018}, \textit{64}, 3.
    \bibitem{G4}V.\ Gheorghiu, B.\ C.\ Sanders, \textit{Phys.\ Rev.\ A} \textbf{2013}, \textit{88}, 022340.
    \bibitem{C7}E.\ Chitambar, D.\ Leung, L.\ Mancinska, M.\ Ozols,  A.\ Winter, \textit{Commun.\ Math.\ Phys.} \textbf{2014}, \textit{328}, 1, 303.
    \bibitem{P3}J.\ Preskill, arXiv:1604.07450.
    \bibitem{K3}M.\ Koashi, N.\ Imoto, \textit{Phys.\ Rev.\ A} \textbf{2002}, \textit{66}, 022318.
    \bibitem{H6}P.\ Hayden, R.\ Jozsa, D.\ Petz,  A.\ Winter, \textit{Commun.\ Math.\ Phys.} \textbf{2004}, \textit{250}, 371.
    \bibitem{K5}R.\ Blume-Kohout, H.\ K.\ Ng, D.\ Poulin, L.\ Viola, \textit{Phys.\ Rev.\ A} \textbf{2010}, \textit{82}, 062306.
    \bibitem{W4}E.\ Wakakuwa, A.\ Soeda, M.\ Murao, \textit{IEEE Trans. Inf. Theory} \textbf{2017}, \textit{63}, 2, 1280.
    \bibitem{L2}H.-K.\ Lo, S.\ Popescu, \textit{Phys.\ Rev.\ A} \textbf{2001}, \textit{63}, 022301.
\end{thebibliography}
\end{document}